\documentclass{lmcs}
\usepackage{amsmath,amssymb}
\usepackage{enumitem}
\usepackage{booktabs,array,longtable}
\usepackage{hyperref}

\newcommand{\N}{\mathbb{N}}
\newcommand{\PC}{\mathcal{PC}}
\newcommand{\Pt}{\widetilde{P}}
\newcommand{\Str}{\mathrm{Str}}
\newcommand{\Exp}{\mathrm{Exp}}
\newcommand{\Lam}{\Lambda}
\newcommand{\emptyfn}{\varnothing}
\newcommand{\Sig}[1]{\Sigma^0_{#1}}
\newcommand{\Pii}[1]{\Pi^0_{#1}}
\newcommand{\Dl}[1]{\Delta^0_{#1}}
\newcommand{\restr}{\!\restriction\!}
\newcommand{\mred}{\le_m}
\newcommand{\Disc}{\mathrm{Disc}}
\newcommand{\Rinst}{\mathcal{R}_{\mathrm{inst}}}
\newcommand{\TOT}{\mathrm{TOT}}
\newcommand{\U}{\mathcal{U}}

\begin{document}
\title[Rice's theorem under self-modification]{Rice's Theorem under Self-Modification: Elevation Operators and a Normal Form}
\author[J.~P.~Gumbau Mezquita]{Jose Pascual Gumbau Mezquita\lmcsorcid{0009-0005-5823-324X}}
\address{Universitat Jaume I, Castell\'o de la Plana, Spain}
\email{gumbau@uji.es}
\subjclass{Theory of computation $\to$ Computability; Theory of computation $\to$ Recursive functions; Theory of computation $\to$ Logic and verification}

\begin{abstract}
We ask whether it can be certified algorithmically that a self-modifying program keeps a behavioural property, a safety property in the motivating case, after its next rewrite (preservation) and along its whole evolution (persistence). When the rewrite depends only on behaviour, preservation is a behavioural property and Rice's theorem applies. When the rewrite reads the code, preservation is no longer behavioural; yet, under a uniform disruption condition, the s-m-n reduction that proves Rice's theorem works inside a single class of behaviourally identical programs, and preservation inherits the degree of the halting problem. One step never exceeds the degree of the property, while persistence can climb one level of the arithmetical hierarchy. We then isolate the mechanism shared by rewriting, supervision and system comparison, the elevation operator, and prove a normal form: the preserving set is determined by a single finite trigger and a polarity, and the Rice-Shapiro theorem restricts the polarity to the arithmetical class of the property. Runtime monitors, consistency supervision, conformance to a reference and observational equivalence are instances, and no sound theory covers the preserving systems.
\end{abstract}
\maketitle

\section{Introduction}\label{sec:intro}

\subsection{The question}
Verifying the safety of a computational system is, formally, deciding whether the system satisfies a semantic property of its behaviour. For fixed systems the limits are classical: Rice's theorem states that every non-trivial semantic property is undecidable. But the systems that motivate this work are not fixed ---they are updated, retrained, rewritten by themselves--- and for them the relevant question is not static but dynamic: not ``is this system safe?'', but ``will this system remain safe once it modifies itself?''. Recursive self-modification, in which a system inspects and rewrites its own code, makes the question unavoidable. The systems that raise it most acutely ---autonomous agents, self-improvement loops, models that retrain themselves--- are instances of the framework studied here; the consequences for their safety verification are developed in a companion applied paper \cite{Gum26}, from which came the intuition that this article makes precise: that the dynamic question is ``Rice one level up'', and that the level up is intensional.

\paragraph{Terminology.} The motivation is the verification of safety, and the word is kept where it is technical: a \emph{safety property} in the sense of Alpern and Schneider, one whose violation is witnessed by a finite event and cannot be undone (Section~\ref{sec:topo}). Everywhere else the results hold for an arbitrary behavioural property $P$, including search-type properties that are not safety properties, and we speak neutrally: $P$ \emph{holds}, is \emph{preserved} across one step, and \emph{persists} along the whole trajectory. Safety is the motivating instance, not the scope.

\subsection{What is proved}\label{sec:proved}
We formalise one step of self-modification as a total computable transformation $\Phi$ on indices, and the dynamic question as the \emph{elevated property} $\Lam_\Phi(P)=\{x\in P:\Phi(x)\in P\}$: ``$P$ holds now and after the step''. The results are of four kinds.

\emph{Closure and engine} (Sections~\ref{sec:operator}--\ref{sec:ext}). When $\Phi$ is extensional ---it depends only on the computed function--- $\Lam_\Phi(P)$ is still a behavioural property and Rice governs it. When $\Phi$ is intensional ---it reads the code--- $\Lam_\Phi(P)$ is no longer behavioural and Rice no longer applies; but under an explicit condition, \emph{uniform disruption} (an inert wrapper that encodes $K$), undecidability does not vanish: the same s-m-n reduction that proves Rice works inside a single behavioural fibre, where Rice cannot look, and $\Lam_\Phi(P)$ inherits the halting degree. We exhibit an effective class of rewrites that do this (syntactic instrumentation), and place the others in an illustrative taxonomy with an upper bound: one elevation step never exceeds the degree of $P$. The two regimes, extensional and intensional, turn out to be the two cases of a single \emph{pullback} of Rice along $\Phi$ (Section~\ref{sec:principle}): Rice always applies to the target, and Kleene's s-m-n theorem is the vehicle that carries the undecidability down to the letter of the original code. One consequence is worth stating from the outset: since $\Phi$ acts at compile time, the dynamic question of one step is a \emph{static} question about the compiled program $\Phi(x)$, and the whole article is also a study of step $0$: of static intensional properties of the form $P\cap\Phi^{-1}(P)$.

\emph{Deepening} (Section~\ref{sec:closure}). Unbounded iteration of the step ---``$P$ holds forever''--- does not relax the barrier but moves it: there are transformations for which $\omega$-persistence is $\Pii2$-complete, and the phenomenon does not need intensionality (the extensional shift $\varphi_x(z+1)$ already reaches it). One step is a single finite event; the limit is infinitely many; and no finite number of iterations leaves the class of $P$.

\emph{Normal form within a class} (Section~\ref{sec:normal}). We isolate the mechanism common to rewriting, supervision and equivalence between systems: the \emph{semantic elevation operator}, a computable mechanism that wraps a base system and reacts to a single finite event ---a $\Sig1$ trigger--- anchored to $K$, with one of two \emph{polarities}: entering the property when the trigger fires (search) or leaving it (safety). For this class we prove a Representation Theorem: the elevated property is $P\cap S_a$ or $P\setminus S_a$, determined by the trigger and the polarity and by nothing else; it inherits $K$ or $\overline K$; the untriggered region is not recursively enumerable. The polarity is not a design choice: the Rice--Shapiro theorem restricts it according to the arithmetical class of $P$ (one can only enter a $\Sig1$ property and only leave a $\Pii1$ one; outside these classes one or both may be possible), and a two-line degree argument confirms both obstructions. The form has a topological explanation ---the trigger is an open set of the trace space, in the sense of Alpern--Schneider--- and the monitor that simulates and fires on violation semidecides it. The class is, moreover, complete: up to the elevated property, every mechanism with an anchored normal form is one of its operators, and the polarity restriction holds for all such mechanisms (Proposition~\ref{prop:complete}). Four structurally different axes ---functional, deductive, conformance to a reference, and monitoring of a forbidden action--- are verified instances, and equivalence between two arbitrary systems ---an operator with its own base--- instantiates the same scheme and recovers the normal form when the trigger is confined; runtime-verification monitors, the G\"odel Machine and Levin search are illustrations.

\emph{Regress} (Section~\ref{sec:sup}). Delegating verification to supervisors does not lower the problem: undecidability belongs to the set, not to the architecture, and a sufficiently general supervisor reproduces it at its own level. A supervisor with bounded resources decides a recursive set and errs on a set of systems that is itself undecidable.

\subsection{What is not claimed}
This article does not claim that \emph{every} computable governance mechanism has the normal form: there are total computable operators outside the class, reacting to infinitely many events, and their study is an open problem. It does not claim that undecidability is always of degree $K$: the anchor requires it by definition, and that every natural mechanism satisfies it is a conjecture. It does not claim that classical results ---program equivalence, incompleteness, Diophantine problems--- are corollaries of the framework: they are at best instances or triggers, and their undecidability is imported. It does not claim that diagonalisation is avoided: it is amortised, done once for the whole class in the anchor axiom. And it does not claim that the class has an intrinsic characterisation: we exhibit $K$-hard intensional operators outside it ---naturally, unbounded iteration and comparison between two arbitrary systems, which keeps the hardness and, when confined, the form; artificially, passive hardness (without an anchor) and mixed polarity (with one)--- and we state the problem of characterising the class fibre by fibre (Section~\ref{sec:open}).

\paragraph{The organising pattern.} All the results instantiate a single pattern, stated in Section~\ref{sec:principle} as a remark and not as a theorem: every computable governance of a process reacts to a finite event; what it can certify is an open set; when the event is anchored to the system's own dynamics, what lies outside the open set is $K$; and, iterated, the pattern climbs the hierarchy. Its value is one of localisation: it says exactly where the total algorithmic guarantee breaks.

\subsection{Structure}
Section~\ref{sec:prelim} fixes the framework, including Rice--Shapiro. Section~\ref{sec:operator} defines the operator, uniform disruption, the instrumentation class and the taxonomy of non-disruptive transformations. Section~\ref{sec:closure} proves the closure theorem and $\Pii2$-completeness. Section~\ref{sec:ext} treats the extensional case. Section~\ref{sec:sup}, the supervision regress. Section~\ref{sec:normal} establishes the general definition, the Representation Theorem, realisability via Rice--Shapiro, the topological reading, the four axes and the illustrations. Section~\ref{sec:principle} proves the pullback lemma, states the organising pattern and proves that no sound theory covers the preserving systems. Section~\ref{sec:open} collects the open problems. Appendix~A places some mechanisms from the literature with respect to the definition.

\section{Preliminaries}\label{sec:prelim}

Fix an acceptable G\"odel numbering $\{\varphi_e\}$ of the partial computable functions \cite{Rog67,Soa87}. We write $\varphi_e(x)\!\downarrow$ if the computation halts and $\varphi_e(x)\!\uparrow$ otherwise, and $\emptyfn$ for the nowhere-defined function. Let $K=\{x:\varphi_x(x)\!\downarrow\}$ be the self-halting set, which is $\Sig1$-complete. We write $T(e,x,s)$ for Kleene's predicate \cite{Kle43} (``$s$ codes a halting computation of $\varphi_e$ on $x$''), so that $\varphi_e(x)\!\downarrow\iff\exists s\,T(e,x,s)$, and $\tau_e(x):=\mu s\,T(e,x,s)$ for the \emph{clock}: the least code of a halting computation. Since $s$ codes the whole computation, input included, $\tau_e(x)\ge x$ whenever $\tau_e(x)<\infty$; hence $\{z:\tau_e(z)<\sigma\}\subseteq\{z<\sigma\}$ is finite for every $\sigma<\infty$. This is a property of the standard coding (the code determines the input), not of the machine model. When we speak of ``steps'' in the sense of resources (Section~\ref{sec:sup}) we say so explicitly; ``time'' always means the clock $\tau$, never physical time.

A set $\Pt$ of partial computable functions is \emph{extensional} (or \emph{behavioural}) if it is closed under equality of functions: $\varphi_a=\varphi_b\Rightarrow(\varphi_a\in\Pt\Leftrightarrow\varphi_b\in\Pt)$. Its \emph{index set} is $P=\{e:\varphi_e\in\Pt\}$. We write $\PC$ for the set of all partial computable functions, and say that $\Pt$ is \emph{non-trivial} if $\emptyset\subsetneq\Pt\subsetneq\PC$.

\begin{thmC}[Rice, \protect\cite{Ric53}]\label{thm:rice}
If $\Pt$ is a non-trivial behavioural property, then $P$ is undecidable.
\end{thmC}

\begin{thmC}[s-m-n, \protect\cite{Kle38}]\label{thm:smn}
There is a total computable function $s$ such that $\varphi_{s(a,x)}(z)=\varphi_a(\langle x,z\rangle)$ for all $a,x,z$.
\end{thmC}

The s-m-n theorem is the tool of every reduction in this article, as it is of the standard proof of Rice's theorem: what changes with respect to Rice is not the engine but \emph{where} the reduction lands (inside a behavioural fibre, Section~\ref{sec:operator}). Kleene's recursion theorem \cite{Kle38} ---for every total computable $f$ there is $e$ with $\varphi_e=\varphi_{f(e)}$--- is not used in any proof; we mention it because it is the point of contact with the Lawvere reading of Section~\ref{sec:cat}.

\begin{thmC}[Rice--Shapiro, \protect\cite{Rog67,Soa87}]\label{thm:RS}
Let $\Pt$ be behavioural with recursively enumerable index set $P$. Then, for every $g\in\PC$, $g\in\Pt$ if and only if there is a finite function $\theta\subseteq g$ with $\theta\in\Pt$.
\end{thmC}

We write the arithmetical hierarchy as $\Sig n,\Pii n$. We use that $\TOT=\{x:\varphi_x\text{ is total}\}$ is $\Pii2$-complete, and Rice--Shapiro to bound from above the complexity of index sets and to restrict the form of mechanisms that react to finite evidence (Section~\ref{sec:normal}).

\paragraph{Program transformations.} A \emph{transformation} is a total computable function $\Phi:\N\to\N$, read as one step of self-modification: $\Phi(x)$ is the index of the program obtained from $x$ by one rewrite. A transformation is \emph{extensional} (semantically well defined) if $\varphi_x=\varphi_y\Rightarrow\varphi_{\Phi(x)}=\varphi_{\Phi(y)}$; otherwise it is \emph{intensional}. This article takes the intensional case as its starting point: $\Phi$ may depend on the source code $x$, not only on the function $\varphi_x$ it computes. This is the situation that models real self-modification, where rewriting operates on program syntax. Unlike the classical undecidability of properties such as confluence or termination in term rewriting systems \cite{Ter03}, which operate over restricted alphabets and rules, here the rewriting is intensional over Turing-complete machines \cite{Tur36}: it transforms the source code of a universal system, not terms of a fixed signature. The extensional case (Section~\ref{sec:ext}) and the intensional one (Sections~\ref{sec:operator}--\ref{sec:closure}) turn out to be the two cases of a single transport of the property along $\Phi$ (pullback lemma, Section~\ref{sec:principle}).

Total computable program rewriters have been classified before from the side of enforcement: Hamlen, Morrisett and Schneider \cite{HMS06} call a policy \emph{RW-enforceable} if some trusted rewriter makes every rewritten program satisfy it while leaving, up to an equivalence, the programs that already do; they show, among other things, that program rewriting adds nothing to static analysis when program equivalence is decidable, and that the RW-enforceable policies coincide with no level of the arithmetical hierarchy. Our question is the dual one: not which properties a trusted rewriter can impose, but whether it can be decided that a given, untrusted rewriter ---the self-modification itself--- preserves a property.

\subsection{Decidability of persistence}
This article works with the classical notion of decidability: a property, identified with its index set, is decidable if that set is recursive. We introduce no notion of ``verifiability'' of our own; the results are statements of (un)decidability, and ``verifiable'' is used only as an informal gloss for ``decidable'', in the sense of a total algorithmic decider ---not in the formal-methods sense of ``provable with human assistance''. The object of the article is not the decidability of a static property but that of its persistence under transformation.

\subsection{The class $\U$}\label{sec:U}
The class $\U$ studied in this article is the decidability face of non-verifiability: properties whose index set is not recursive. The resource face (intractability) follows the same unfolding-plus-trigger architecture in a bounded regime, and an instance is given in companion work \cite{Gum26}.

So that the closure theorem is not empty ---every undecidable property is undecidable--- $\U$ is defined by its provenance: $\U$ is the family of non-trivial index sets of behavioural properties, together with their one-step elevations $\Lam_\Phi(P)$ (Section~\ref{sec:operator}) under \emph{uniformly disruptive} transformations $\Phi$ (Definition~\ref{def:unif}, in either polarity). It is not a closure in the strict sense: Definition~\ref{def:unif} is formulated for behavioural $P$, and $\Lam_\Phi(P)$ is not behavioural, so the elevation does not iterate inside $\U$ (iteration is the object of $\Lam^\omega_\Phi$, Section~\ref{sec:closure}). The restriction to disruptive $\Phi$ is necessary: without it a constant $\Phi$ with value outside $P$ would give $\Lam_\Phi(P)=\emptyset$, which is decidable. With it, membership $\Lam_\Phi(P)\in\U$ is immediate by construction, and the substantive content of the closure theorem (Theorem~\ref{thm:closure}) is another: \emph{every member of $\U$ inherits the degree of $K$ or of $\overline K$}, by the same s-m-n reduction as Rice, but inside a single behavioural fibre, where Rice cannot look. The weight of the definition thus shifts to the non-vacuity and effective characterisation of the class of uniformly disruptive $\Phi$ ($\Rinst$, Theorem~\ref{thm:rinst}), and Section~\ref{sec:normal} shows that, for the mechanisms of a precisely defined class (one finite event, one polarity, a fixed witness), the members of $\U$ have a form fixed by the trigger and the polarity.

\paragraph{Dynamics as step $0$.} It is worth fixing from now a reading that Section~\ref{sec:principle} develops. $\Phi$ is total computable and acts at compile time: it reads $x$ and outputs $\Phi(x)$. Hence $\Lam_\Phi(P)=P\cap\Phi^{-1}(P)$ is a \emph{static} set of indices: the question ``will $P$ still hold after the step?'' is the question ``does the program $\Phi(x)$ satisfy $P$?'', a step-$0$ property of the compiled code which, read on $x$, is intensional. The dynamics of one step adds no temporal dimension; it absorbs it into a single finite event in the execution of $\varphi_{\Phi(x)}$. The only dynamics that is not absorbed is unbounded iteration (Section~\ref{sec:closure}). Thus everything this article says about self-modification is also a statement about static intensional properties ---the class $\U$ is a class of non-behavioural index sets--- and the results hold for any question of the form $P\cap\Phi^{-1}(P)$, whether $\Phi$ is read as a rewrite, a compilation, an instrumentation or a translation.

\subsection{Terminology and notation}\label{sec:gloss}

The article uses many terms with a fixed meaning. We collect them, with a one-line definition and the place where they are introduced, so that none has to be inferred from context.

{\footnotesize
\begin{longtable}{@{}>{\raggedright\arraybackslash}p{3.0cm}>{\raggedright\arraybackslash}p{2.3cm}>{\raggedright\arraybackslash}p{6.6cm}>{\raggedright\arraybackslash}p{2.0cm}@{}}
\caption{Terminology and notation.}\label{tab:gloss}\\
\toprule
Term & Symbol & One-line definition & Where \\
\midrule
\endfirsthead
\multicolumn{4}{@{}l}{\emph{Table~\ref{tab:gloss} (continued)}}\\
\toprule
Term & Symbol & One-line definition & Where \\
\midrule
\endhead
\bottomrule
\endlastfoot
behavioural property & $\Pt$, $P$ & set of functions closed under equality; $P$ its index set. ``Semantic'' is an informal synonym & Section~\ref{sec:prelim} \\
syntactic / intensional & --- & depends on the code $x$, not only on $\varphi_x$. ``Extensional'' = depends only on $\varphi_x$ & Section~\ref{sec:prelim} \\
transformation & $\Phi$ & total computable function on indices: \emph{one step} of rewriting. ``Rewrite'' is the informal reading of $\Phi$ & Section~\ref{sec:prelim} \\
one-step elevation & $\Lam_\Phi(P)$ & $\{x\in P:\Phi(x)\in P\}$: ``$P$ holds now and after one step'' & Def.~\ref{def:one} \\
limit elevation ($\omega$-persistence) & $\Lam^\omega_\Phi(P)$ & $P$ holds along the whole trajectory & Def.~\ref{def:omega} \\
fibre & $\{x:\varphi_x=\varphi_t\}$ & class of all indices computing the same function & Section~\ref{sec:operator} \\
(inert) wrapper & $w(e)$ & program that computes $\varphi_t$ and carries $e$ as data in its code; ``inert'' = $\varphi_{w(e)}=\varphi_t$ & Lemma~\ref{lem:inert} \\
extractor & $\delta$ & reads $e$ from the code: $\delta(w(e))=e$ & Def.~\ref{def:rinst} \\
uniform disruption & --- & $\Phi$ with an inert wrapper and $\Phi(w(e))\in P\iff e\in K$ (positive) or $\iff e\notin K$ (negative) & Def.~\ref{def:unif} \\
instrumentation class & $\Rinst(P)$ & transformations extractor + wrapper + synthesiser & Def.~\ref{def:rinst} \\
trigger & $\Str_a(x)$, $S_a$, $\sigma_a(x)$ & $\varphi_a(x)\!\downarrow$; its set (r.e.); the clock at which it fires. The trigger \emph{fires}; a monitor \emph{raises an alarm} & Section~\ref{sec:def} \\
semantic elevation operator & $E$ & total computable function with trigger, polarity, base and anchor & Def.~\ref{def:oes} \\
polarity A (search) & --- & the output \emph{enters} $P$ when the trigger fires & Def.~\ref{def:oes} \\
polarity B (safety) & --- & the output \emph{leaves} $P$ when the trigger fires & Def.~\ref{def:oes} \\
base & $g$ (A), $f$ (B) & function ``before'' (A, $\varphi_g\subseteq\varphi_t$) or ``after'' (B) the trigger & Def.~\ref{def:oes} \\
mix & $\mathrm{mix}_\sigma(t,f)$ & $\varphi_t$ on the inputs finished before $\sigma$, $\varphi_f$ on the rest & Def.~\ref{def:oes} \\
anchor & (W1), (W2) & inert wrapper with $\Str_a(w(e))\iff e\in K$ & Def.~\ref{def:oes} \\
untriggered region & $\Exp(E)$ & $\{x:\neg\Str_a(x)\}$; complement of $S_a$ & Theorem~\ref{thm:main} \\
preserving set & $\Lam_E(P)$ & under polarity B, $P\setminus S_a$, i.e.\ $P\cap\Exp(E)$: the systems that satisfy $P$ and keep it under $E$ & Cor.~\ref{cor:exp} \\
finite event / finite evidence & $[\eta]$ & a finite trace prefix $\eta$; the cylinder containing it; equivalent to $\Sig1$ & Section~\ref{sec:topo} \\
monitor & $M$ & procedure that reads the trace incrementally and may raise an alarm & Prop.~\ref{prop:mon} \\
axis / instance & ob; $(\mathrm{ob},Q,t,a,$ $\pi,\text{base})$ & observation of behaviour that defines $\Pt$; tuple satisfying the definition & Def.~\ref{def:eix} \\
\midrule
compile time & --- & when $\Phi$ or $E$ compute the output index; always halts & Section~\ref{sec:nonvac}, Section~\ref{sec:def} \\
run time & --- & the execution of $\varphi_{E(x)}(z)$; the unbounded search lives here & Section~\ref{sec:def} \\
execute & --- & make a program compute (on the universal machine, or as a subcomputation of another); never ``run'' in any other sense & Section~\ref{sec:prelim} \\
simulate & --- & synonym of execute when the executor is another program or a monitor; used for readability, no content of its own & Section~\ref{sec:nonvac} \\
clock & $\tau_e(x)$ & least code of a halting computation ($\ge x$); ``time'' always means this & Section~\ref{sec:prelim} \\
step (of rewriting) & $\Phi$, $\Phi^k$ & one application of $\Phi$; ``$k$ steps'' = $k$ applications & Def.~\ref{def:one} \\
steps (of machine) & --- & computational resources, only in Section~\ref{sec:sup} (bounded deciders) & Section~\ref{sec:sup} \\
\end{longtable}
}

\section{The Semantic Elevation Operator Iterates Functionally}\label{sec:operator}

The operator $\Lam_\Phi$ of this section is the semantic elevation operator on its functional axis: it elevates the property $P$ to the property ``$P$ is preserved under one step of $\Phi$'', and iterates by composition. Section~\ref{sec:normal} gives the general definition, of which the instrumentation class $\Rinst$ (Section~\ref{sec:rinst}) is the first literal instance; every uniformly disruptive $\Phi$ coincides with an instance on its wrappers.

\begin{defi}[one-step elevation]\label{def:one}
For a property $P$ (an index set) and a transformation $\Phi$, the \emph{one-step elevation} of $P$ under $\Phi$ is
\[
\Lam_\Phi(P):=\{x\in\N : x\in P\wedge\Phi(x)\in P\}.
\]
\end{defi}
In words: $x$ satisfies the elevated property if $P$ holds now ($x\in P$) and still holds after the next rewrite ($\Phi(x)\in P$). We adopt the conjunction ``now and after'' because it composes cleanly under iteration.

\begin{defi}[limit elevation, $\omega$-persistence]\label{def:omega}
The \emph{limit elevation} of $P$ under $\Phi$ is
\[
\Lam^\omega_\Phi(P):=\{x\in\N : \forall k\ge0,\ \Phi^k(x)\in P\},
\]
where $\Phi^k$ is the $k$-th iterate of $\Phi$ ($\Phi^0=\mathrm{id}$).
\end{defi}
In words: no future rewriting trajectory, of any length, breaks $P$. This is the formalisation of persistent alignment.

\begin{prop}[coinductive characterisation]\label{prop:coind}
Let $F_P(S)=P\cap\Phi^{-1}(S)=\{x\in P:\Phi(x)\in S\}$, and define $\Lam^{[k]}_\Phi(P):=F_P^k(\N)$ (the descending chain from $\top=\N$). Then $\Lam^{[k]}_\Phi(P)=\{x:\forall j<k,\ \Phi^j(x)\in P\}$; $F_P$ is co-continuous, hence $\Lam^\omega_\Phi(P)=\bigcap_k\Lam^{[k]}_\Phi(P)=\mathrm{gfp}(F_P)$, the greatest fixed point of $F_P$.
\end{prop}

\begin{rem}[classical Tarski vs.\ effective realisability]
The Tarski--Knaster theorem guarantees that the gfp exists as a set, and co-continuity guarantees that it is reached in exactly $\omega$ steps. This is purely set-theoretic: $\Lam^\omega_\Phi(P)$ is well defined. Theorem~\ref{thm:pi2} shows that, for certain $\Phi$, this well-defined object is not computable: it exists classically but no procedure decides it. This tension has a natural categorical formulation in the effective topos, sketched in Section~\ref{sec:cat}.
\end{rem}

\begin{rem}[intensionality]
Since $\Phi$ is intensional, the function $\varphi_{\Phi(x)}$ is not determined by $\varphi_x$: two indices $x,y$ with $\varphi_x=\varphi_y$ may have $\varphi_{\Phi(x)}\ne\varphi_{\Phi(y)}$. Hence $\Lam_\Phi(P)$ is not, in general, a behavioural property, and Rice does not apply to it: its undecidability requires a proof of its own.
\end{rem}

\subsection{The disruption condition}\label{sec:disr}
For $\Lam_\Phi(P)$ to be undecidable, the disruption of $\Phi$ cannot be an accident of two isolated indices: the reduction from $K$ must build, for each instance $e$, a witness whose fate under $\Phi$ encodes whether $e\in K$.

\begin{defi}[uniform disruption via wrappers]\label{def:unif}
Let $\Pt$ be behavioural and non-trivial and $t$ an index with $\varphi_t\in\Pt$. A transformation $\Phi$ is \emph{uniformly disruptive} with respect to $\Pt$ if there is a total computable function $w:\N\to\N$ (the \emph{wrapper}) such that, for every index $e$:
\begin{enumerate}[label=(\roman*),leftmargin=2.2em]
\item $\varphi_{w(e)}=\varphi_t$ (the wrapper never alters the computed function); and
\item $\Phi(w(e))\in P\iff e\in K$ (\emph{positive} disruption), or $\Phi(w(e))\in P\iff e\notin K$ (\emph{negative} disruption).
\end{enumerate}
\end{defi}

Condition (i) makes the disruption genuinely intensional: all the $w(e)$ compute the same function $\varphi_t$, and yet $\Phi$ sends them into or out of $P$ according to a fact ---the halting of $\varphi_e(e)$--- that depends not on the computed function but on the syntax of the wrapper. So no uniformly disruptive $\Phi$ can be extensional, and the article cannot collapse into a corollary of Rice. Crucially, $w$ is total computable: it always builds the wrapper, without ever deciding whether $e\in K$.

\begin{rem}[the two polarities]
Positive disruption makes the rewrite \emph{enter} $P$ when $\varphi_e(e)$ halts; negative disruption makes it \emph{leave}. Both are needed: for $P\in\Pii1$ ---every safety property--- the positive one is impossible, because $K\mred P$ via $\Phi\circ w$ would make $K$ co-r.e.; for $P\in\Sig1$ the negative one is, by the same argument. Section~\ref{sec:normal} shows that Rice--Shapiro explains \emph{why} each one fails.
\end{rem}

\begin{lem}[uniform reduction inside a fibre]\label{lem:red}
If $\Phi$ is uniformly disruptive with respect to $\Pt$, then $K\mred\Lam_\Phi(P)$ (positive) or $\overline K\mred\Lam_\Phi(P)$ (negative), via $w$. Consequently $\Lam_\Phi(P)$ is undecidable.
\end{lem}

\begin{proof}
By (i), $\varphi_{w(e)}=\varphi_t\in\Pt$ for every $e$, so $w(e)\in P$ unconditionally. By (ii), $\Phi(w(e))\in P$ iff $e\in K$ (resp.\ $e\notin K$). Hence $w(e)\in\Lam_\Phi(P)\iff e\in K$ (resp.\ $e\notin K$). Since $K$ and $\overline K$ are not recursive and $\mred$ preserves non-recursiveness, $\Lam_\Phi(P)$ is not recursive. The reduction is not circular: $w(e)$ never decides $e\in K$; it only builds a wrapper carrying $e$ as data, and it is the (non-)halting of the simulation of $\varphi_e(e)$ ---inside $\Phi(w(e))$, not inside the construction of $w$--- that determines membership in $\Lam_\Phi(P)$.
\end{proof}

\subsection{Non-vacuity and a clarifying contrast}\label{sec:nonvac}
Fix $\Pt=\{\mathrm{id}\}$, the property ``computes the identity function''; note that $\emptyfn\notin\Pt$. By the s-m-n theorem, define a total $w$ such that $\varphi_{w(e)}$ is the program that returns $z$ on input $z$, carrying $e$ embedded as inert syntactic data; then $\varphi_{w(e)}=\mathrm{id}\in\Pt$ for every $e$. Define $\Phi$ as a total computable syntactic rewriter: given $x$, if $x$ does not have the structure of a wrapper $w(e)$, let $\Phi(x)=x$; if $x=w(e)$, let $\Phi(w(e))$ be the index, generated by s-m-n, of a program $y$ that executes the simulation of $\varphi_e(e)$ and, if it halts, returns $z$. Then $\varphi_{\Phi(w(e))}=\mathrm{id}$ if $e\in K$, and $=\emptyfn$ if $e\notin K$: this is condition (ii) of Definition~\ref{def:unif}. $\Phi$ does not execute $\varphi_e(e)$: it is a pure rewriter that reads the syntax of $w(e)$, extracts $e$ and synthesises the code of $y$ by s-m-n ---a compile-time operation, always terminating. The simulation lives in $\varphi_{\Phi(w(e))}$, at run time.

\paragraph{A contrast that does not work.} Consider a $\Phi'$ that rewrites according to a decidable syntactic mark, e.g.\ the parity of the index. This $\Phi'$ is intensional and pointwise disruptive, but not uniformly disruptive: to use it in a reduction from $K$ one would have to produce wrappers of the parity dictated by $e\in K$ ---that is, decide $K$. Intensionality by itself is not enough; it must be able to encode an undecidable fact, and it is the embedded simulation that achieves this.

\subsection{An effective class of disruptive rewrites}\label{sec:rinst}

\begin{defi}[syntactic instrumentation class $\Rinst(P)$]\label{def:rinst}
Let $\Pt$ be non-trivial with $\emptyfn\notin\Pt$, and $t$ an index with $\varphi_t\in\Pt$. A total computable $\Phi$ belongs to $\Rinst(P)$ if it decomposes into three compatible total computable functions: an \emph{extractor} $\delta$ that reads an index $x$ and extracts a parameter $e=\delta(x)$; a \emph{wrapper generator} $w$ with $\delta(w(e))=e$ and $\varphi_{w(e)}=\varphi_t$; and a \emph{conditional synthesiser} $\mathrm{inj}$ (by s-m-n) such that $\Phi(x)=\mathrm{inj}(t,\delta(x))$, where $\varphi_{\mathrm{inj}(t,e)}(z)$ is computed by executing $\varphi_e(e)$ and, if it halts, returning $\varphi_t(z)$; that is,
\[
\varphi_{\mathrm{inj}(t,e)}=\varphi_t\ \text{if } e\in K,\qquad \varphi_{\mathrm{inj}(t,e)}=\emptyfn\ \text{if } e\notin K.
\]
\end{defi}

\begin{rem}[why $\emptyfn$ and not an arbitrary witness]
A synthesiser that produced exactly $\varphi_f$, for an arbitrary $f\notin P$, when $\varphi_e(e)\!\uparrow$ is not realisable: the program can never know that the simulation will not halt, and it cannot retract the answers it has already given. What it can produce when the search does not halt is any $\varphi_g\subseteq\varphi_t$: it answers $\mathrm{dom}\,\varphi_g$ at once and the rest only after the search. Lemma~\ref{rem:BB} gives the exact statement: a switch between two functions, in either direction, is realisable if and only if the ``before'' function is a subfunction of the ``after'' one. $\Rinst$ is defined with $g=\emptyfn$ for simplicity; Section~\ref{sec:normal} generalises it to any base $g\subseteq t$ with $\varphi_g\notin\Pt$. For properties with $\emptyfn\in\Pt$ the negative polarity is needed, which Section~\ref{sec:normal} realises.
\end{rem}

\begin{lem}[inert wrapper]\label{lem:inert}
For every index $t$ there are total computable functions $w$ and $\delta$ with $\varphi_{w(e)}=\varphi_t$ and $\delta(w(e))=e$ for every $e$.
\end{lem}
\begin{proof}
By s-m-n with padding, $w$ can be taken injective with decidable range (the s-m-n function of an acceptable numbering admits this form, \cite[\S1.8]{Rog67}): $w(e)$ is the index of the program that ignores $e$ and executes $t$, and $\delta$, total computable, reads $e$ from the code when the index is in the range of $w$ and takes a fixed value otherwise.
\end{proof}

This is the only construction the anchor axiom of Section~\ref{sec:normal} needs: given any trigger of the form $\varphi_{\delta(x)}(\delta(x))\!\downarrow$, the inertness and coupling conditions hold automatically. The work in each instance is not building the wrapper, but choosing the polarity and the base that the property admits.

\begin{thm}[every syntactic instrumentation is disruptive]\label{thm:rinst}
Every $\Phi\in\Rinst(P)$ is uniformly disruptive with respect to $\Pt$.
\end{thm}

\begin{proof}
By the condition on $w$, $\varphi_{w(e)}=\varphi_t\in\Pt$ for every $e$ ---condition (i). By compatibility $\Phi(w(e))=\mathrm{inj}(t,e)$; if $e\in K$, $\varphi_{\Phi(w(e))}=\varphi_t\in\Pt$; if $e\notin K$, $\varphi_{\Phi(w(e))}=\emptyfn\notin\Pt$. So $\Phi(w(e))\in P\iff e\in K$, condition (ii). By Lemma~\ref{lem:red}, $K\mred\Lam_\Phi(P)$.
\end{proof}

\subsection{What non-disruptive transformations do}\label{sec:nondisr}

Uniform disruption is not the only source of undecidability of $\Lam_\Phi(P)$; it is the source of $K$-hardness through a single finite event. It is worth placing it among the others. For every total computable $\Phi$, $\Lam_\Phi(P)=P\cap\Phi^{-1}(P)$ (Section~\ref{sec:principle}, pullback lemma), so the question is always what degree $\Phi^{-1}(P)$ has.

\begin{prop}[one-step upper bound]\label{prop:upper}
For every total computable $\Phi$ and every $P$, we have $\Phi^{-1}(P)\mred P$ via $\Phi$, and hence $\Lam_\Phi(P)\le_T P$. One elevation step never rises above the degree of $P$.
\end{prop}

\begin{proof}
$x\in\Phi^{-1}(P)\iff\Phi(x)\in P$, and $\Phi$ is total computable. $\Lam_\Phi(P)$ is the intersection of $P$ with a set $\mred P$.
\end{proof}

\begin{rem}[taxonomy]\label{rem:tax}
Six illustrative families of transformations ---not exhaustive, as the example at the end shows---, ordered by the relation between $\Phi^{-1}(P)$ and $P$:
\begin{enumerate}[label=(\arabic*),leftmargin=2.2em]
\item \emph{Constants} ($\Phi(x)=c$): $\Phi^{-1}(P)\in\{\N,\emptyset\}$, and $\Lam_\Phi(P)\in\{P,\emptyset\}$. No dynamic content.
\item \emph{Extensional} (Section~\ref{sec:ext}): $\Phi^{-1}(P)$ is behavioural; $\Lam_\Phi(P)$ is undecidable by Rice if non-trivial. No intensionality.
\item \emph{Decidable trigger} (Section~\ref{sec:nonvac}, parity): $\Lam_\Phi(P)=P\cap R$ with $R$ recursive. It may be undecidable or decidable depending on $P$ and on the numbering (there are acceptable numberings in which the even indices all index $\emptyfn$); in any case the inheritance is passive: it comes from $P$, not from $\Phi$.
\item \emph{Anchor at an intermediate degree} (Remark~\ref{rem:weak}): $\Phi^{-1}(P)$ non-recursive and not $K$-hard, by Friedberg--Muchnik.
\item \emph{Reaction to infinitely many events} (Remark~\ref{rem:fora}): $\Phi^{-1}(\TOT)=\{x:W_{\delta(x)}\text{ infinite}\}$ with $\delta$ a syntactic extractor, intensional and $\Pii2$-complete. Harder than $K$ in a single step ---but not harder than $\TOT$, by Proposition~\ref{prop:upper}. This family \emph{is} uniformly disruptive (with $h$ such that $W_{h(e)}=\N$ if $e\in K$ and $\emptyset$ otherwise, $w\circ h$ is an inert wrapper that encodes $K$); it lies outside $\mathcal{E}(P)$, not outside Definition~\ref{def:unif}.
\item \emph{Sanitisers} ($\Phi(x)\in P$ for every $x$, e.g.\ a rewriter that always outputs the witness $t$): $\Phi^{-1}(P)=\N$ and $\Lam_\Phi(P)=P$. This is the only family in which rewriting removes the dynamic question, at the price of not preserving the system's function.
\end{enumerate}
The families do not exhaust the transformations: a $\Phi$ that is the identity outside the range of a wrapper and of form (A) on it, with an anchor to an r.e.\ set of intermediate degree, with $P=\TOT$, has $\Phi^{-1}(P)$ $\Pii2$-complete and falls in none of the six. What Definition~\ref{def:unif} isolates is the mechanism ``one finite event anchored to the dynamics of the code itself'', which is what real systems do and what Section~\ref{sec:normal} studies. The general classification of transformations by the degree of $\Phi^{-1}(P)$ relative to $P$ remains open (Section~\ref{sec:open}). Proposition~\ref{prop:upper} closes from above what uniform disruption closes from below, and explains why the jump to $\Pii2$ in Theorem~\ref{thm:pi2} is only possible by iteration: $P_0$ is $\Sig1$, and no isolated step can leave it.
\end{rem}

\paragraph{Scope of $\Rinst$ (what we claim and what we do not).} $\Rinst(P)$ is a natural and effective class of intensional rewrites. We do not claim that every computable rewrite lies in $\Rinst$: function-preserving optimisations without instrumentation, refactorings, and transformations that change the computed function do not. $\Rinst$ is an existential witness, not a totality: it establishes that there are effective and realistic rewrites under which persistence is undecidable. What does follow from Theorem~\ref{thm:rinst} is that, \emph{given} an instrumented system $w(e)$, deciding whether the rewrite keeps it in $P$ is equivalent to deciding $K$. Whether it is decidable if an \emph{arbitrary} $\Phi$ is of the disruptive kind is a different question, which we leave open (Section~\ref{sec:open}).

\section{The Closure Theorem}\label{sec:closure}

\begin{thm}[closure under one-step elevation]\label{thm:closure}
Let $P$ be the index set of a non-trivial behavioural property, and let $\Phi$ be uniformly disruptive with respect to $\Pt$. Then $\Lam_\Phi(P)$ is undecidable and, more precisely, $K\mred\Lam_\Phi(P)$ or $\overline K\mred\Lam_\Phi(P)$. Moreover, $\Lam_\Phi(P)$ is not a behavioural property.
\end{thm}

\begin{proof}
The reduction is Lemma~\ref{lem:red}. For non-behaviourality, take $e\in K$ and $e'\notin K$: by (i) $\varphi_{w(e)}=\varphi_{w(e')}$, and by (ii) exactly one of the two lies in $\Lam_\Phi(P)$.
\end{proof}

\begin{rem}
Since $\U$ is defined as the family of elevations under such $\Phi$ (\S\ref{sec:U}), membership $\Lam_\Phi(P)\in\U$ is immediate by construction; the content of the theorem is that every member of $\U$ inherits the degree of $K$ or of $\overline K$, and \emph{how} it does so: once extensionality is lifted, Rice no longer applies ---$\Lam_\Phi(P)$ is not behavioural--- but the same s-m-n reduction that proves Rice keeps working, now inside a single behavioural fibre $\{x:\varphi_x=\varphi_t\}$. The engine does not change; the scale at which it looks does.
\end{rem}

\begin{cor}[application to instrumentation rewrites]\label{cor:rinst}
For every $\Phi\in\Rinst(P)$, $K\mred\Lam_\Phi(P)$: there are effective and realistic self-modifications under which the preservation of $P$ is not decidable. Since these rewrites preserve the function computed by the instrumented system, no decider can distinguish the benign cases from the dangerous ones without deciding $K$.
\end{cor}

\begin{cor}[the expressivity barrier is stable]\label{cor:stable}
For every uniformly disruptive $\Phi$, the set $\Lam_\Phi(P)$ belongs to $\U$ and inherits $K$ or~$\overline K$; non-verifiability does not vanish in passing from the static to the dynamic question, but is preserved. Under uniform disruption, non-verifiability is a fixed point of elevation, not an artefact of the static case.
\end{cor}

\paragraph{Two independent sources.} The undecidability of Theorem~\ref{thm:closure} comes from intensionality: the reduction lands inside a behavioural fibre, where Rice cannot look. The arithmetical jump that follows has a different, independent origin: unbounded iteration along the rewriting trajectory. That it is independent is shown by the fact that it \emph{does not need intensionality}: the shift $\varphi_{\Phi(x)}(z)=\varphi_x(z+1)$ is extensional and gives $\Lam^\omega_\Phi(P_0)=\TOT$ exactly, which is $\Pii2$-complete. The $\Phi$ of Theorem~\ref{thm:pi2} is an intensional version of the same phenomenon ---and is, moreover, uniformly disruptive (Remark~\ref{rem:pi2disr})---, so the two sources can be combined; but neither reduces to the other.

\begin{thm}[limit elevation and $\Pii2$-completeness]\label{thm:pi2}
Let $P_0=\{e:\varphi_e(0)\!\downarrow\}\in\Sig1$. There is a total computable transformation $\Phi$ such that $\Lam^\omega_\Phi(P_0)$ is $\Pii2$-complete.
\end{thm}

\begin{proof}
\emph{Upper bound.} With Kleene's predicate $T$, $e\in P_0\iff\exists s\,T(e,0,s)$. By Definition~\ref{def:omega}, $x\in\Lam^\omega_\Phi(P_0)\iff\forall k\,\Phi^k(x)\in P_0\iff\forall k\,\exists s\,T(\Phi^k(x),0,s)$, which is $\Pii2$.

\emph{Lower bound.} We reduce $\TOT$, which is $\Pii2$-complete. By s-m-n define a total $w:\N\times\N\to\N$ with $\varphi_{w(e,k)}(z)=0$ if $\varphi_e(k)\!\downarrow$ and $\uparrow$ otherwise, so that $w(e,k)\in P_0\iff\varphi_e(k)\!\downarrow$. Let $\delta$ recognise $w(e,k)$, and set $\Phi(x)=w(e,k+1)$ if $\delta(x)=(e,k)$, otherwise $\Phi(x)=x$. With $g(e)=w(e,0)$, the iterates are $\Phi^k(g(e))=w(e,k)$, so $g(e)\in\Lam^\omega_\Phi(P_0)\iff\forall k\,\varphi_e(k)\!\downarrow\iff e\in\TOT$. So $\TOT\mred\Lam^\omega_\Phi(P_0)$; combined with the upper bound, $\Lam^\omega_\Phi(P_0)$ is $\Pii2$-complete.
\end{proof}

\begin{rem}[the $\Phi$ of the theorem is disruptive; an extensional variant]\label{rem:pi2disr}
The $\Phi$ above is intensional (the $w(e,k)$ with $\varphi_e(k)\!\downarrow$ all compute the constant $0$ and are sent to images computing different functions) and, in fact, uniformly disruptive with respect to $\Pt_0$: by s-m-n let $h$ satisfy $\varphi_{h(e)}(0)=0$ and $\varphi_{h(e)}(1)\simeq\varphi_e(e)$, and $w'(e):=w(h(e),0)$; all the $w'(e)$ compute $\lambda z.0\in\Pt_0$, and $\Phi(w'(e))=w(h(e),1)\in P_0\iff e\in K$. So the same transformation witnesses Theorem~\ref{thm:closure} (one step) and Theorem~\ref{thm:pi2} (limit). The extensional variant $\varphi_{\Phi'(x)}(z)=\varphi_x(z+1)$ gives $\Lam^\omega_{\Phi'}(P_0)=\TOT$ and shows that the $\Pii2$ jump is a phenomenon of iteration, not of intensionality.
\end{rem}

\begin{rem}[degrees of undecidability]
The jump to $\Pii2$-completeness by limit iteration belongs to the line of extensions of the recursion theorem to higher degrees: Arslanov's completeness criterion \cite{Ars77} characterises the incomplete r.e.\ sets as those for which the fixed-point theorem holds (below $K$ there are fixed-point-free functions). We do not use it; we mention it because it places the jump within degree theory, not because the jump is intensional: Remark~\ref{rem:pi2disr} shows it is due to iteration. In the terms of Section~\ref{sec:normal}, one elevation step is a single finite event ($\Sig1$); unbounded iteration adds a universal quantifier in front ---infinitely many events--- and no normal form with a single trigger captures it (Remark~\ref{rem:fora}).
\end{rem}

\begin{cor}[the expressivity barrier deepens]\label{cor:deep}
There are transformations and $\Sig1$ properties for which there is no point in the rewriting chain at which verification becomes decidable: iterating the elevation moves the barrier up the arithmetical hierarchy, to $\Pii2$-completeness. (For $\Pii1$ properties the limit stays at $\Pii1$: the iteration adds a $\forall$, which is absorbed.) Certifying persistent alignment is, for these transformations, strictly harder than certifying $P$ at one instant, and for two independent reasons: the intensionality of the individual step (Theorem~\ref{thm:closure}) and the unbounded iteration of the trajectory (Theorem~\ref{thm:pi2}).
\end{cor}

\section{The Extensional Case, by Contrast}\label{sec:ext}

We close the map with the extensional regime. Here the result is a corollary of Rice; we prove it for completeness and to make the dividing line with the intensional case visible.

\begin{defi}
Let $\Phi$ be extensional and $\Pt$ behavioural and non-trivial. We say that $\Pt$ is \emph{$P$-disruptive} under $\Phi$ if there are $x,y\in P$ with $\Phi(x)\in P$ and $\Phi(y)\notin P$; that is, if the elevated property is non-trivial inside $P$.
\end{defi}

\begin{prop}[extensional case]\label{prop:ext}
If $\Phi$ is extensional and $\Pt$ is a non-trivial behavioural property that is $P$-disruptive under $\Phi$, then $\Lam_\Phi(P)$ is the index set of a non-trivial behavioural property, and hence undecidable by Rice.
\end{prop}

\begin{proof}
Since $\Phi$ is extensional, $\varphi_x=\varphi_y\Rightarrow\varphi_{\Phi(x)}=\varphi_{\Phi(y)}$, so $\Phi$ descends to a well-defined $\widehat\Phi$ on behaviours: $\widehat\Phi(g):=\varphi_{\Phi(x)}$ for any $x$ with $\varphi_x=g$. Then $Q:=\{g\in\PC : g\in\Pt\wedge\widehat\Phi(g)\in\Pt\}$ is a behavioural property, and its index set is $\Lam_\Phi(P)$. $Q$ is non-trivial by $P$-disruption ($\varphi_x\in Q$, $\varphi_y\notin Q$), so Rice applies directly.
\end{proof}

The difference from Theorem~\ref{thm:closure} is exactly one of transport: in both cases $\Lam_\Phi(P)=P\cap\Phi^{-1}(P)$, and Rice always applies to the target $P$. If $\Phi$ is extensional, $\Phi^{-1}(P)$ remains behavioural (via the descent $\widehat\Phi$) and Rice is enough; if $\Phi$ looks at the letter of the code, the descent no longer exists, $\Phi^{-1}(P)$ ceases to be behavioural, and the s-m-n reduction has to be carried out by hand, inside a behavioural fibre, to carry the undecidability of the target down to the original code. This places the extensional/intensional line not as the frontier between two theorems but as the two cases of a single pullback (Lemma~\ref{lem:pullback}) made with the same engine: below it, Rice does everything; above it, Rice's reduction has to be repeated at a scale where Rice does not see. The Expressivity Principle holds on both sides; what changes is the scale.

\section{The Supervision Regress Does Not Terminate}\label{sec:sup}

A natural way to avoid the undecidability of preservation is to delegate it: if no decider can certify $M_0$, perhaps a supervisor $M_1$, at least as capable, can verify it; an $M_2$ can verify $M_1$; and so on. This section shows that delegation adds nothing: the undecidability of $\Lam_\Phi(P)$ is a property of the set, not of the architecture that tries to decide it, and no tower of supervisors ---finite, composite, adaptive--- lowers it. What the tower does do is reproduce the problem at each level: a sufficiently general supervisor is itself a system of the class studied, and its own preservation is equally undecidable.

\subsection{Supervisors and certificates}

\begin{defi}[supervisor, certificate]\label{def:sup}
Let $M_0$ be a system with transition operator $\Phi_0$ and safety property $P$. A \emph{supervisor} of $M_0$ is a computable procedure $V$ that, on an index $x$, outputs a verdict $V(x)\in\{\mathsf{yes},\mathsf{no},\uparrow\}$. We say that $V$ gives a \emph{total correct certificate of preservation} if $V$ is total and $V(x)=\mathsf{yes}\iff x\in\Lam_{\Phi_0}(P)$. A \emph{tower} $M_0,M_1,\dots,M_k$ is a sequence of supervisors in which each $M_{n+1}$ may consult the verdicts of $M_n$ (and, by composition, of all lower levels); the certificate of the tower is the verdict of the top level.
\end{defi}

\begin{thm}[absolute undecidability; non-termination of the regress]\label{thm:abs}
Let $\Lam_{\Phi_0}(P)$ be undecidable (by Proposition~\ref{prop:ext} if $\Phi_0$ is extensional, by Theorem~\ref{thm:closure} if it is intensional and uniformly disruptive). Then no finite tower $M_0,M_1,\dots,M_k$ gives a total correct certificate of preservation, whatever the architecture, capability or adaptivity of its levels.
\end{thm}

\begin{proof}
The verdict of level $k$ is the output of a computable procedure that consults other computable procedures; a composition of computable procedures is computable. If this verdict were a total correct certificate, it would be a total computable procedure deciding $\Lam_{\Phi_0}(P)$, against the hypothesis. No property of the architecture enters: the undecidability belongs to the set.
\end{proof}

\begin{rem}
The theorem is deliberately an observation: its strength is that it does not depend on \emph{how} supervision is done. More elaborate versions ---supervisors that simulate, inspect code, adjust their criteria according to what they observe--- are special cases, and no additional hypothesis about them adds or removes anything. What remains to be seen is what a \emph{concrete} supervisor can and cannot do, and what happens to the supervisor itself when it is sufficiently general.
\end{rem}

\subsection{Total supervisors: where they err}

\begin{lem}[total deciders decide recursive sets]\label{lem:bounded}
Let $V$ be a procedure that halts on every index $x$. Then $V$ decides a recursive set. In particular, if $V$ halts within a constant number $c$ of machine steps, it reads at most a prefix of length $c$ of the code of $x$ and decides a set determined by that prefix.
\end{lem}

\begin{proof}
$V$ is total by hypothesis and its verdict is computable; the set $\{x:V(x)=\mathsf{yes}\}$ is recursive. If $V$ halts within $c$ steps, it can have read at most $c$ symbols of the code, so $V(x)$ depends only on the prefix of length $c$ of $x$.
\end{proof}

\begin{prop}[the error of a total supervisor cannot be located]\label{prop:err}
Let $\Lam_{\Phi_0}(P)$ be non-recursive and let $V$ be a total supervisor, with set of positive verdicts $R=\{x:V(x)=\mathsf{yes}\}$. Then the error set $\mathrm{Err}(V):=R\,\triangle\,\Lam_{\Phi_0}(P)$ is infinite and non-recursive: $V$ errs on infinitely many systems, and the set of systems on which it errs is itself undecidable.
\end{prop}

\begin{proof}
$R$ is recursive by the lemma, so $R\ne\Lam_{\Phi_0}(P)$ and $\mathrm{Err}(V)\ne\emptyset$. If $\mathrm{Err}(V)$ were recursive, $\Lam_{\Phi_0}(P)=R\,\triangle\,\mathrm{Err}(V)$ would be recursive. A non-recursive set is infinite.
\end{proof}

\begin{rem}[where the error lies]
Under uniform disruption, Proposition~\ref{prop:err} can be made explicit: the reduction $K\mred\Lam_{\Phi_0}(P)$ (or $\overline K\mred$, depending on the polarity) via the wrapper $w$ (Lemma~\ref{lem:red}) shows that $V$ must err on infinitely many $w(e)$, and that the set $\{e: V \text{ errs on } w(e)\}$ is non-recursive. No enlargement of $V$'s budget ---more steps, more memory, more levels--- changes this: it only moves the errors, never removes them nor makes them recognisable. This is the precise form of the ``finite horizon'' intuition: it is not that the supervisor sees \emph{little}, but that what it fails to see is not enumerable.
\end{rem}

\subsection{The supervisor as a system of the class}

Theorem~\ref{thm:abs} says that the tower does not lower the problem of level $0$. We now show that it reproduces it: a general supervisor is a system of the same class, and the question ``does the supervisor go on supervising correctly?'' has the same complexity as the original question.

\begin{prop}[the supervisor's own preservation]\label{prop:tower}
\begin{enumerate}[label=(\roman*),leftmargin=2.2em]
\item Let $M_{n+1}$ be a supervisor that is itself a system whose transition operator $\Phi_{n+1}$ is uniformly disruptive with respect to a non-trivial property $\Pt_{n+1}$ (for instance, ``issues correct verdicts on a fixed family''). Then $\Lam_{\Phi_{n+1}}(P_{n+1})$ is undecidable (Theorem~\ref{thm:closure}), and $K$ or $\overline K$ reduces to it depending on the polarity.
\item The paradigmatic case is the supervision of consistency (Proposition~\ref{prop:ded}): the operator $E_{\mathrm{ded}}$ that assigns to a deductive system $x$ its supervisor belongs to $\mathcal{E}_{\mathrm B}(P_{\mathrm{con}})$, and the preservation of consistency under supervision is $\Pii1$-complete.
\end{enumerate}
\end{prop}

\begin{proof}
(i) is Theorem~\ref{thm:closure} applied at level $n+1$. (ii) is Proposition~\ref{prop:ded}.
\end{proof}

\begin{cor}[the two faces of the tower]\label{cor:faces}
The regress does not terminate either for extensional or for intensional supervisors. If all the $\Phi_n$ are extensional, each $\Lam_{\Phi_n}(P_n)$ is undecidable by Proposition~\ref{prop:ext} (Rice); if they are intensional and uniformly disruptive, by Theorem~\ref{thm:closure} (s-m-n inside a fibre). In neither case does the tower converge: it is a chain of undecidable problems in which no level decides the previous one.
\end{cor}

\begin{rem}[two elevations, one form]
The hierarchical elevation of this section (stacking supervisors) and the functional elevation of Sections~\ref{sec:operator}--\ref{sec:closure} (iterating $\Phi$) have different structures and should not be identified: the first changes system at each level, the second evolves a single system. Nonetheless, Section~\ref{sec:normal} shows that the mechanism that injects undecidability is the same in both: a wrapper that reacts to a single finite event anchored to $K$. The consistency supervisor of Proposition~\ref{prop:ded} is, literally, an elevation operator of polarity B; the self-rewriting system of Section~\ref{sec:rinst} is one of polarity A. The tower is not an escape route but another instance of the class.
\end{rem}

\begin{rem}[what this section does not say]
We do not claim that an adequate supervisor \emph{must} be Turing-complete or self-modifying: Definition~\ref{def:sup} does not presuppose it, and Theorem~\ref{thm:abs} holds for every supervisor. We claim that, if it is, it inherits the problem (Proposition~\ref{prop:tower}); and that, if it is total, it decides a recursive set and errs in a way that cannot be located (Proposition~\ref{prop:err}). Neither option closes the regress.
\end{rem}

\section{The Normal Form of Semantic Elevation}\label{sec:normal}

The closure theorem (Section~\ref{sec:closure}) establishes that, under uniform disruption, the elevated property $\Lam_\Phi(P)$ inherits the degree of $K$; the supervision regress (Section~\ref{sec:sup}) shows that this inheritance is not avoided by delegating. In both cases the mechanism is the same: a wrapper that executes a subcomputation and reacts to a finite event. This section isolates this mechanism as an object in its own right ---the \emph{semantic elevation operator}--- and proves four things about it: that the elevated property depends only on the trigger and the polarity, not on the wrapper (normal form); that the polarity is restricted by the arithmetical class of the property (Rice--Shapiro); that the untriggered region is not even recursively enumerable; and that four structurally different axes ---functional, deductive, conformance and monitoring--- are literal instances of the same definition, with a catalogue showing how the arithmetical class of the property constrains the polarity.

We do not claim that \emph{every} computable mechanism of semantic governance has this form. We claim three more modest and provable things: that, within the class we define, the elevated property is determined by the trigger and the polarity and by nothing else; that the polarity is restricted by the arithmetical class of $P$; and that the form is the one that s-m-n and Rice--Shapiro allow to a procedure that can only react to finite evidence. The definition \emph{stipulates} the form, it does not derive it: there are computable mechanisms that react to a single finite event and do not belong to it (Remark~\ref{rem:fora}), and those that depend on infinitely many events are those that Section~\ref{sec:closure} obtains by iteration.

\paragraph{What the operator is for.} Three things, in order of importance. \emph{Unification}: rewriting a program, supervising a deductive system, conforming to a reference and watching for a forbidden action are, under the definition, the same object with different observations (Section~\ref{sec:axes}), and what is proved about it holds for all at once. \emph{Reading off complexity}: once the anchor and the polarity hypothesis are verified for a concrete mechanism, the elevated property, its hardness, its class and the non-enumerability of the untriggered region are read off Theorem~\ref{thm:main} with no further construction (usage guide, Remark~\ref{rem:guia}). \emph{Fixing the direction}: the definition makes visible that a mechanism can only make systems enter or leave, and Rice--Shapiro says which of the two each property allows (Section~\ref{sec:real}); this is what a designer needs to know before building anything.

\subsection{Definition}\label{sec:def}

Fix a non-trivial behavioural property $\Pt$, with index set $P$, and a witness $t$ with $\varphi_t\in\Pt$. We write $\emptyfn$ for the nowhere-defined function. For an index $a$ we define the \emph{trigger}
\[
\Str_a(x) \;:\Longleftrightarrow\; \varphi_a(x)\!\downarrow \;\Longleftrightarrow\; \exists s\, T(a,x,s),
\]
and $\sigma_a(x) := \mu s\,T(a,x,s)$ (with $\sigma_a(x)=\infty$ if $\Str_a(x)$ fails). The set $S_a := \{x : \Str_a(x)\}$ is r.e.\ by construction, and every r.e.\ set is of the form $S_a$ for some $a$ (Kleene's normal form theorem). Recall that the clock $\tau_t(z)=\mu s\,T(t,z,s)$ is a computation code and satisfies $\tau_t(z)\ge z$, so that $\{z:\tau_t(z)<\sigma\}$ is finite for every $\sigma<\infty$ (Section~\ref{sec:prelim}).

\begin{defi}[semantic elevation operator]\label{def:oes}
A \emph{semantic elevation operator over $P$} is a total computable function $E:\N\to\N$ for which there are a trigger $a$, a polarity $\pi\in\{\mathrm{A},\mathrm{B}\}$ and a \emph{base} index ($g$ if $\pi=\mathrm{A}$, with $\varphi_g\subseteq\varphi_t$; $f$ if $\pi=\mathrm{B}$), such that for every $x$:
\begin{itemize}[leftmargin=2em]
\item[(A)] \emph{Search polarity.} $\varphi_{E(x)} = \varphi_t$ if $\Str_a(x)$, and $\varphi_{E(x)}=\varphi_g$ otherwise. (The case $g=\emptyfn$ is that of $\Rinst$.)
\item[(B)] \emph{Safety polarity.} For every input $z$, $\varphi_{E(x)}(z)$ is computed by executing $\varphi_t(z)$ and $\varphi_a(x)$ in parallel: if $\varphi_t(z)$ halts before $\varphi_a(x)$, it returns $\varphi_t(z)$; otherwise it returns $\varphi_f(z)$. Thus $\varphi_{E(x)}=\varphi_t$ if $\neg\Str_a(x)$, and if $\Str_a(x)$ with $\sigma=\sigma_a(x)$,
\[
\varphi_{E(x)} \;=\; \mathrm{mix}_\sigma(t,f) \;:=\; \varphi_t\restr\{z:\tau_t(z)<\sigma\}\;\cup\;\varphi_f\restr\{z:\tau_t(z)\ge\sigma\}.
\]
\end{itemize}
and which satisfies the \emph{anchor axiom}: there is a total computable function $w:\N\to\N$ (the \emph{wrapper}) such that, for every $e$,
\begin{itemize}[leftmargin=2em]
\item[(W1)] \emph{inertness:} $\varphi_{w(e)}=\varphi_t$;
\item[(W2)] \emph{coupling:} $\Str_a(w(e)) \iff e\in K$.
\end{itemize}
The \emph{polarity hypotheses} on $(\Pt,t,g)$ and $(\Pt,t,f)$ are:
\[
(\mathrm{H_A})\;\; \varphi_g\notin\Pt; \qquad\qquad
(\mathrm{H_B})\;\; \mathrm{mix}_\sigma(t,f)\notin\Pt \text{ for every } \sigma<\infty.
\]
\end{defi}

By s-m-n, for each $(t,a,\pi,g)$ or $(t,a,\pi,f)$ there is an operator realising (A) or (B) uniformly in $x$ (for (A): on input $z$ it executes $\varphi_g(z)$ and $\varphi_a(x)$ in parallel; if $\varphi_g(z)$ halts, it returns its value, which is $\varphi_t(z)$; if $\varphi_a(x)$ halts, it computes and returns $\varphi_t(z)$; this is exact because $\varphi_g\subseteq\varphi_t$, Lemma~\ref{rem:BB}); we call it \emph{canonical} and denote it $E_{t,a,g}^{\mathrm{A}}$ (resp.\ $E_{t,a,f}^{\mathrm{B}}$). The definition, however, admits any $E$ producing the same functions: the wrapper may do whatever it likes with $x$ as long as the result is that of (A) or (B). The clause is \emph{for every $x$} and with a single polarity: an operator that were (A) on a decidable region and (B) on another does not belong, even if it reacts to a single event and has an anchor (Remark~\ref{rem:fora}). An operator defined only on the wrappers extends to the class by setting $E(x):=E^{\pi}_{t,a,\cdot}(x)$ outside the range of $w$.

\begin{cor}[expressivity criterion: (a)+(b')+(c)]\label{cor:abc}
Every semantic elevation operator satisfies:
\begin{itemize}[leftmargin=2em]
\item[(a)] \emph{Computability:} $E$ is total computable; it acts at compile time and always halts. The unbounded search lives in $\varphi_{E(x)}$, not in $E$.
\item[(b')] \emph{Active coupling:} there is a family $\{w(e)\}$ of programs of constant semantics $\varphi_t$ such that $E(w(e))\in P\iff e\in K$ (polarity A) or $\iff e\notin K$ (polarity B). In particular $E$ is intensional: it distinguishes indices that compute the same function, and the distinction encodes $K$.
\item[(c)] \emph{Structural respect:} the canonical operator realises $\varphi_{E(x)}$ as a wrapper that executes $\varphi_a(x)$ and $\varphi_t$ as subcomputations, without altering their logic.
\end{itemize}
\end{cor}
\begin{proof}
(a) is Definition~\ref{def:oes}. (b'): by (W1) $\varphi_{w(e)}=\varphi_t$; by (W2) and the polarity clause, $\varphi_{E(w(e))}=\varphi_t$ exactly when $e\in K$ (A) or $e\notin K$ (B), and otherwise $\varphi_g$ or $\mathrm{mix}_\sigma(t,f)$, which under $(\mathrm{H}_\pi)$ are not in $\Pt$. (c) is the s-m-n construction of the canonical operator.
\end{proof}

\begin{rem}[the degenerate cases]
(b') excludes the two cases of Section~\ref{sec:nonvac}: a constant $E$ does not vary with $e$; a reaction to a decidable mark (parity) varies, but aligning it with $e\in K$ would require $w$ to decide $K$. (c) is a property of the canonical realisation, not an axiom: the definition fixes only the computed function, and $E$ may do whatever it likes with $x$ as long as the result is that of (A) or (B).
\end{rem}

\begin{rem}[why (c) cannot be an axiom]\label{rem:noc}
There are three candidate formal readings of ``$E(x)$ wraps $x$ as its base engine'', and none does the work of an axiom. \emph{(c-ext)}: there is a functional $\Gamma$ with $\varphi_{E(x)}=\Gamma(\varphi_x)$; this is incompatible with (b'), because on the inert family $\varphi_{E(w(e))}=\Gamma(\varphi_t)$ would be constant in $e$ (Corollary~\ref{cor:ext}(i)) ---and it is the only reading under which Myhill--Shepherdson would apply, which rules out applying it to any member of the class. \emph{(c-int)}: there is $\Gamma$ with $\varphi_{E(x)}=\Gamma(\varphi_x,x)$; this is empty, because $\Gamma(g,x)(z):=\varphi_{E(x)}(z)$ satisfies it for every total computable $E$. \emph{(c-syn)}: $E(x)=s(u,x)$ for a fixed index $u$ (s-m-n uniformity); this adds nothing, because for every $E$ there is $E'$ of this form with $\varphi_{E'(x)}=\varphi_{E(x)}$ and, since $P$ is an index set, $\Lam_{E'}(P)=\Lam_E(P)$. This is why (c) is formulated on the canonical realisation and the definition is limited to the computed function.
\end{rem}

\begin{lem}[irrevocability and realisations of polarity B]\label{rem:BB}
Let $a$ have $S_a$ non-recursive, and let $t,f$ be indices.
\begin{enumerate}[label=(\roman*),leftmargin=2.2em]
\item (\emph{Irrevocability.}) There is a total computable $E$ with $\varphi_{E(x)}=\varphi_t$ if $\neg\Str_a(x)$ and $\varphi_{E(x)}=\varphi_f$ if $\Str_a(x)$ (exact switch) if and only if $\varphi_t\subseteq\varphi_f$. Symmetrically, the switch from $\varphi_g$ (before) to $\varphi_t$ (after) is realisable if and only if $\varphi_g\subseteq\varphi_t$: this is polarity A of Definition~\ref{def:oes}.
\item (\emph{Augmentation.}) If $\varphi_f=\varphi_t\cup\theta$ with $\theta$ finite and $\mathrm{dom}\,\theta\cap\mathrm{dom}\,\varphi_t=\emptyset$, then $\mathrm{mix}_\sigma(t,f)=\varphi_t\cup\theta$ for every $\sigma$, and $(\mathrm{H_B})$ reduces to $\varphi_t\cup\theta\notin\Pt$.
\item (\emph{Truncation.}) If $\varphi_f=\emptyfn$, then $\mathrm{mix}_\sigma(t,\emptyfn)=\varphi_t\restr\{z:\tau_t(z)<\sigma\}$, a finite restriction of $\varphi_t$; $(\mathrm{H_B})$ follows from no finite restriction of $\varphi_t$ being in $\Pt$ (and only needs the restrictions $\varphi_t\restr\{\tau_t<\sigma\}$).
\end{enumerate}
\end{lem}
\begin{proof}
(i) If $\varphi_t\subseteq\varphi_f$, the race of (B) realises the exact switch: on $z\in\mathrm{dom}\,\varphi_t$ it returns $\varphi_t(z)=\varphi_f(z)$ in either case, and on $z\notin\mathrm{dom}\,\varphi_t$ it returns $\varphi_f(z)$ iff $\Str_a(x)$. Conversely, if $\varphi_t\not\subseteq\varphi_f$ there is $z_0$ with $\varphi_t(z_0)\!\downarrow$ and $\varphi_f(z_0)$ undefined or different. Suppose $E$ realises the exact switch. If $\varphi_f(z_0)\!\uparrow$, then $\varphi_{E(x)}(z_0)\!\downarrow\iff\neg\Str_a(x)$, and the complement of $S_a$ would be r.e., whence $S_a$ recursive. If $\varphi_f(z_0)\!\downarrow\ne\varphi_t(z_0)$, then $\varphi_{E(x)}(z_0)$ is computable in $x$ (it always halts) and equals $\varphi_t(z_0)$ iff $\neg\Str_a(x)$: $S_a$ would be recursive. Contradiction in both cases.
(ii) and (iii) are direct computations from the definition of $\mathrm{mix}_\sigma$.
\end{proof}

\begin{rem}[reading the lemma]
The program has already returned $\varphi_t(z)$ for the inputs finished before the trigger and cannot retract them; it can only change what it has not yet answered. This irrevocability is the only form that the ``continuity'' of every computable reaction takes inside the definition. Truncation serves for properties that are lost by \emph{ceasing} to answer (totality, computing a fixed function); augmentation, for properties that are lost by \emph{doing} something forbidden (safety in the strict sense). The hypothesis ``$S_a$ non-recursive'' in the lemma is stated for clarity: for a semantic elevation operator it follows from the anchor axiom (Theorem~\ref{thm:main}(i)), and the lemma applies to the class with no additional assumption. One caveat to keep in mind: $\mathrm{mix}_\sigma(t,f)$ depends on the index $t$ and on the clock $\tau$, not only on the function $\varphi_t$; hence $(\mathrm{H_B})$ is a hypothesis on $(\Pt,t,\tau)$ and the class $\mathcal{E}_{\mathrm B}(P)$ depends on the computation model through the clock. The conclusions of Theorem~\ref{thm:main} ---complexity and hardness--- do not.
\end{rem}

\subsection{Hardness, normal form and untriggered region}\label{sec:hard}

\begin{thm}[Representation Theorem: normal form, hardness and untriggered region]\label{thm:main}
Let $E$ be a semantic elevation operator over $P$ with trigger $a$, polarity $\pi$ and wrapper $w$, and assume $(\mathrm{H}_\pi)$. Write $\Exp(E):=\N\setminus S_a=\{x:\neg\Str_a(x)\}$ (the \emph{untriggered region}). Then:
\begin{enumerate}[label=(\roman*),leftmargin=2.2em]
\item $S_a$ is $\Sig1$-complete and $\Exp(E)$ is $\Pii1$-complete; in particular $\Exp(E)$ is not r.e.
\item (\emph{Normal form.}) $E^{-1}(P)=S_a$ if $\pi=\mathrm{A}$, and $E^{-1}(P)=\Exp(E)$ if $\pi=\mathrm{B}$. Hence
\[
\Lam_E(P) := \{x\in P : E(x)\in P\} \;=\; P\cap S_a \quad(\mathrm{A}), \qquad \Lam_E(P) \;=\; P\setminus S_a \quad(\mathrm{B}),
\]
and the elevated property depends only on $(a,\pi)$, not on the wrapper nor on the concrete realisation of $E$.
\item (\emph{Hardness.}) $K\mred\Lam_E(P)$ if $\pi=\mathrm{A}$, and $\overline{K}\mred\Lam_E(P)$ if $\pi=\mathrm{B}$, in both cases via $w$. Consequently $\Lam_E(P)$ is not recursive.
\item (\emph{Upper bound.}) If $P\in\Sig1$ and $\pi=\mathrm{A}$, $\Lam_E(P)$ is $\Sig1$-complete; if $P\in\Pii1$ and $\pi=\mathrm{B}$, $\Lam_E(P)$ is $\Pii1$-complete. In general, if $P$ belongs to a class $\Gamma$ of the hierarchy closed under intersection with $\Sig1$ (resp.\ $\Pii1$), then $\Lam_E(P)\in\Gamma$.
\item $\Lam_E(P)$ is not a behavioural property: there are $x,y$ with $\varphi_x=\varphi_y$, $x\in\Lam_E(P)$ and $y\notin\Lam_E(P)$. Rice does not apply; the undecidability of (iii) comes from $K$ by s-m-n inside the fibre of $t$.
\end{enumerate}
\end{thm}

\begin{proof}
(i) $S_a$ is r.e.\ by construction. By (W2), $e\in K\iff w(e)\in S_a$, so $K\mred S_a$ via $w$; $K$ being $\Sig1$-complete, so is $S_a$. Taking complements, $\Exp(E)$ is $\Pii1$-complete. If $\Exp(E)$ were r.e., $S_a$ would be recursive, and so would $K$.

(ii) Polarity A: if $\Str_a(x)$, $\varphi_{E(x)}=\varphi_t\in\Pt$; if not, $\varphi_{E(x)}=\varphi_g\notin\Pt$ by $(\mathrm{H_A})$. Polarity B: if $\neg\Str_a(x)$, $\varphi_{E(x)}=\varphi_t\in\Pt$; if $\Str_a(x)$, $\varphi_{E(x)}=\mathrm{mix}_\sigma(t,f)\notin\Pt$ by $(\mathrm{H_B})$. The formula for $\Lam_E(P)$ is the definition. That it depends only on $(a,\pi)$ is immediate from the formula.

(iii) By (W1), $w(e)\in P$ for every $e$. Polarity A: $w(e)\in\Lam_E(P)\iff w(e)\in S_a\iff e\in K$. Polarity B: $w(e)\in\Lam_E(P)\iff w(e)\notin S_a\iff e\notin K$. Since $\mred$ preserves non-recursiveness, $\Lam_E(P)$ is not recursive.

(iv) By (ii), $\Lam_E(P)$ is the intersection of $P$ with a $\Sig1$ set (A) or a $\Pii1$ set (B); the indicated classes are closed under these intersections, and completeness follows from (iii).

(v) Take $e\in K$ and $e'\notin K$. By (W1), $\varphi_{w(e)}=\varphi_{w(e')}=\varphi_t$; by (iii), exactly one of $w(e)$, $w(e')$ lies in $\Lam_E(P)$.
\end{proof}

\begin{cor}[neither the preserving set nor the untriggered region is r.e.]\label{cor:exp}
Under (B), the set of systems that preserve $P$ under $E$ is $P\setminus S_a$, and it is not r.e.\ (nor is the untriggered region $\Exp(E)$, by Theorem~\ref{thm:main}(i)): there is no procedure that semidecides membership ---no enumeration, not even a partial one, of the preserving systems. Under (A), $\Lam_E(P)=P\cap S_a$ is r.e.\ if $P$ is, but never recursive: success can be certified, failure never.
\end{cor}
\begin{proof}
If $P\setminus S_a$ were r.e., since $\overline K\mred P\setminus S_a$ via $w$, $\overline K$ would be r.e.\ and $K$ recursive. The rest is Theorem~\ref{thm:main}(ii)--(iv).
\end{proof}

\begin{defi}[the classes $\mathcal{E}_{\mathrm A}(P)$ and $\mathcal{E}_{\mathrm B}(P)$]\label{def:classes}
We denote by $\mathcal{E}_{\mathrm A}(P)$ (resp.\ $\mathcal{E}_{\mathrm B}(P)$) the class of semantic elevation operators over $P$ of polarity A (resp.\ B) satisfying $(\mathrm{H}_\pi)$, and by $\mathcal{E}(P)$ their union. Theorem~\ref{thm:main}(ii) says that the map
\[
\mathcal{E}(P)\;\longrightarrow\;\{\text{elevated properties}\},\qquad E\longmapsto\Lam_E(P)\in\{P\cap S_a,\;P\setminus S_a\},
\]
factors through $E\mapsto(a,\pi)$: this is the \emph{invariant} of the class. Two operators with the same trigger and polarity are indistinguishable by the property they elevate, however different their wrappers.
\end{defi}

\begin{rem}[where uniqueness lives]
Point (ii) of the theorem, read through Definition~\ref{def:classes}, is all the content that the expression ``normal form'' has in this framework: the diversity of governance mechanisms is a diversity of wrappers; the elevated property sees only $(a,\pi)$.
\end{rem}

\begin{defi}[anchored normal form]\label{def:anf}
Let $E$ be a total computable function and $\pi\in\{\mathrm A,\mathrm B\}$. We say that $E$ has an \emph{anchored normal form of polarity $\pi$ over $P$} if there are a $\Sig1$ set $S$, an index $t$ with $\varphi_t\in\Pt$ and a total computable $w$ with $\varphi_{w(e)}=\varphi_t$ for every $e$, such that $\Lam_E(P)=P\cap S$ (if $\pi=\mathrm A$) or $\Lam_E(P)=P\setminus S$ (if $\pi=\mathrm B$), and $w(e)\in S\iff e\in K$ for every $e$.
\end{defi}

\begin{prop}[the class is complete for anchored normal forms]\label{prop:complete}
Let $E$ be any total computable function with an anchored normal form of polarity $\pi$ over $P$, with witness $t$ and family $w$. Suppose that $P$ admits polarity $\pi$ at $t$: there is $g$ with $\varphi_g\subseteq\varphi_t$ and $\varphi_g\notin\Pt$ (if $\pi=\mathrm A$), or there is $f$ such that $(t,f)$ satisfies $(\mathrm{H_B})$ (if $\pi=\mathrm B$). Then there is $E'\in\mathcal{E}_\pi(P)$, anchored by the same $w$, with $\Lam_{E'}(P)=\Lam_E(P)$.
\end{prop}
\begin{proof}
By Kleene's normal form there is $a$ with $S_a=S$. Let $E'$ be the canonical operator $E^{\pi}_{t,a,g}$ (resp.\ $E^{\pi}_{t,a,f}$). Condition (W1) is $\varphi_{w(e)}=\varphi_t$; condition (W2) is $\Str_a(w(e))\iff w(e)\in S\iff e\in K$; and $(\mathrm{H}_\pi)$ is the hypothesis. Hence $E'\in\mathcal{E}_\pi(P)$, and by Theorem~\ref{thm:main}(ii) its elevated property is $P\cap S_a$ or $P\setminus S_a$, which is $\Lam_E(P)$.
\end{proof}

\begin{rem}[what this settles, and what it does not]\label{rem:complete}
Proposition~\ref{prop:complete} says that Definition~\ref{def:oes} is not a choice: up to the elevated property, every mechanism with an anchored normal form is one of the canonical operators, so the class $\mathcal{E}(P)$ is derived from the phenomenon rather than stipulated. Moreover, the degree argument in the proof of Proposition~\ref{prop:RS} uses only the anchor and the polarity, and therefore applies to every mechanism with an anchored normal form: if $\pi=\mathrm A$, then $w(e)\in P$ and $E(w(e))\in P\iff w(e)\in\Lam_E(P)\iff e\in K$, so $K\mred P$ and $P\notin\Pii1$; if $\pi=\mathrm B$, symmetrically $\overline K\mred P$ and $P\notin\Sig1$. No mechanism whatsoever with an anchored normal form enters a $\Pii1$ property or leaves a $\Sig1$ one. What remains open is to characterise anchored normal forms by the action of $E$ itself, fibre by fibre (Problem~\ref{prob:inert}).
\end{rem}

\begin{rem}[what lies outside the class]\label{rem:fora}
There are total computable, intensional, $K$-hard operators that do not belong to $\mathcal{E}(P)$, of three kinds. \emph{Infinitely many events.} Let $\delta$ be a syntactic extractor and $E$ such that $\varphi_{E(x)}(z)=0$ if there is $y\ge z$ with $y\in W_{\delta(x)}$ (search), and undefined otherwise. Then $\varphi_{E(x)}$ is total iff $W_{\delta(x)}$ is infinite, so $E^{-1}(\TOT)=\{x:W_{\delta(x)}\text{ infinite}\}$, a $\Pii2$-complete set (for $\delta$ with $\delta(w(y))=y$ on a generator $w$, $\{y:W_y\text{ infinite}\}$ reduces to it via $w$), and $E$ is intensional (two indices of the same function with different $\delta$). No $\Sig1$ trigger determines it: by (ii), every $E\in\mathcal{E}(\TOT)$ has $E^{-1}(\TOT)$ in $\Sig1$ or in $\Pii1$. The reaction depends on a $\forall$ in front of the $\exists$, and it is exactly what Section~\ref{sec:closure} obtains by iteration ($\Lam^\omega_\Phi$, $\Pii2$-completeness). \emph{One event, non-uniform polarity.} For $P$ admitting both polarities with base $\emptyfn$ ($\TOT$, $\{\mathrm{id}\}$), the operator that is (A) on a decidable region and (B) on another (Section~\ref{sec:open}) has an anchor and reacts to a single event, but $E^{-1}(P)$ is a symmetric difference of $S_a$ with a recursive set, neither $\Sig1$ nor $\Pii1$. \emph{Own base.} The relational operator $E_{\mathrm{rel}}$ (Section~\ref{sec:axes}) has a wrapper and an anchor and inherits the hardness (iii), but it is not a literal instance: outside the trigger it returns the system itself, not a fixed $\varphi_t$, so $E^{-1}(P)\ne\Exp(E)$. The form (ii) holds for it when the trigger is confined to systems on which $(\mathrm{H_B})$ holds pointwise, and fails otherwise (Proposition~\ref{prop:rel} and the following remark). The class $\mathcal{E}(P)$ is thus that of operators reacting to \emph{one} finite event, \emph{with a single polarity} and \emph{with a fixed witness}; none of the three conditions follows from having an anchor, and the third is the one that an operator with its own base can recover by confinement. The relation to Theorems~\ref{thm:closure} and~\ref{thm:pi2} is: one step, one event; the limit, infinitely many. The study of operators outside the class is left to Section~\ref{sec:open}.
\end{rem}

\subsection{Realisability: the polarity against the grain of the class}\label{sec:real}

Recall the Rice--Shapiro theorem (Section~\ref{sec:prelim}): if the index set $P$ of $\Pt$ is r.e., then for every $g\in\PC$,
$g\in\Pt \iff \exists\,\theta\subseteq g$ finite with $\theta\in\Pt$.

\begin{prop}[realisability: Rice--Shapiro and polarity]\label{prop:RS}
Let $\Pt$ be non-trivial.
\begin{enumerate}[label=(\roman*),leftmargin=2.2em]
\item If $P\in\Pii1$, then every partial computable subfunction of an element of $\Pt$ is in $\Pt$; in particular $\emptyfn\in\Pt$ and, for every $g$ with $\varphi_g\subseteq\varphi_t$, $\varphi_g\in\Pt$: $(\mathrm{H_A})$ fails for every base and no operator of polarity A elevates $P$.
\item If $P\in\Sig1$, then for every $t$ with $\varphi_t\in\Pt$ and every $f$, $\mathrm{mix}_\sigma(t,f)\in\Pt$ for $\sigma$ large enough: $(\mathrm{H_B})$ fails and no operator of polarity B elevates $P$.
\item There are properties $\Pt$ with $P\in\Pii2\setminus(\Sig1\cup\Pii1)$ for which both polarities are realisable; $\TOT$ and $\{\mathrm{id}\}$ are examples. Consequently, $\mathcal{E}_{\mathrm A}(P)=\emptyset$ if $P\in\Pii1$, $\mathcal{E}_{\mathrm B}(P)=\emptyset$ if $P\in\Sig1$, and both classes are non-empty for $\TOT$ and $\{\mathrm{id}\}$.
\end{enumerate}
\end{prop}

\begin{proof}
(i) The complement $\Pt^c$ has r.e.\ index set; by Rice--Shapiro, $h\in\Pt^c$ iff some finite restriction of $h$ is in $\Pt^c$, and hence $\Pt^c$ is upward closed under extension: if $\theta\in\Pt^c$ and $\theta\subseteq h$, then $h\in\Pt^c$. Contrapositively, $\Pt$ is downward closed: every partial computable subfunction of an element of $\Pt$ is in $\Pt$. With $\varphi_t\in\Pt$, every $\varphi_g\subseteq\varphi_t$ is in $\Pt$. (For $\emptyfn$: if $\emptyfn\in\Pt^c$, then $\Pt^c=\PC$, against non-triviality.)

(ii) By Rice--Shapiro applied to $\varphi_t\in\Pt$, there is a finite $\theta\subseteq\varphi_t$ with $\theta\in\Pt$, and every extension of $\theta$ is in $\Pt$. Let $\sigma_0:=1+\max\{\tau_t(z):z\in\mathrm{dom}\,\theta\}$. For $\sigma\ge\sigma_0$, $\mathrm{mix}_\sigma(t,f)\supseteq\theta$, so $\mathrm{mix}_\sigma(t,f)\in\Pt$.

(iii) $\TOT$: $\emptyfn$ is not total ($\mathrm{H_A}$ with $g=\emptyfn$); no finite restriction of a total function is total (truncation, $\mathrm{H_B}$). $\{\mathrm{id}\}$: likewise. Both index sets are $\Pii2$-complete, in particular outside $\Sig1\cup\Pii1$. The example $\Pt=\TOT\cup\{\emptyfn\}$ shows why the base $g$ of (A) cannot always be $\emptyfn$: here $\emptyfn\in\Pt$, but with $g$ a non-empty finite restriction of $\varphi_t$ one has $\varphi_g\notin\Pt$ and (A) is realisable.

\emph{Second proof of (i) and (ii), by degrees.} If $E\in\mathcal{E}_{\mathrm A}(P)$ with $P\in\Pii1$, by (W2) and the polarity $e\in K\iff E(w(e))\in P$: $K\mred P$ via $E\circ w$, and $K$ would be $\Pii1$, hence recursive. If $E\in\mathcal{E}_{\mathrm B}(P)$ with $P\in\Sig1$, $e\notin K\iff E(w(e))\in P$: $\overline K\mred P$ and $\overline K$ would be r.e., hence $K$ recursive. Neither uses Rice--Shapiro; what Rice--Shapiro adds is \emph{which hypothesis fails} ---$(\mathrm{H_A})$ or $(\mathrm{H_B})$--- and hence that the obstruction is not accidental but one of form: a finite trigger can only make a system enter a property that finite evidence certifies, and only make it leave one that finite evidence refutes.
\end{proof}

\begin{rem}[reading]
A trigger is finite evidence. A $\Sig1$ property is, by Rice--Shapiro, one in which finite evidence certifies \emph{membership}; it can only be \emph{entered} through a trigger (polarity A). A $\Pii1$ property is one in which finite evidence certifies \emph{non-membership}; it can only be \emph{left} through a trigger (polarity B). The polarity is not a design choice: it runs against the grain of the class. For properties above $\Sig1\cup\Pii1$ the restriction no longer comes from Rice--Shapiro, and what remains depends on the property: both directions are open for $\TOT$ and $\{\mathrm{id}\}$, but not always ---the $\Pii2$-complete property $\{g : g\subseteq 0\restr K\}$ is downward closed, so $(\mathrm{H_A})$ fails for every base and only B is admitted. This is the precise form of the intuition ``search vs.\ safety'' (\emph{liveness} vs.\ \emph{safety}) within this framework. Its counterpart for trace properties is the fact that the violation of a safety property, and the satisfaction of a co-safety property, are witnessed by a finite prefix \cite{AS85}. Monitorability in the sense of Pnueli and Zaks is strictly broader: a trace property is monitorable exactly when its boundary in the Cantor topology has empty interior, which holds for every set in $G_\delta\cap F_\sigma$, including properties that are neither safety nor co-safety (Diekert and Leucker \cite{DL14}; see \cite{BLS11} for monitor constructions). Its natural counterpart here would be operators reacting to several finite events (Problem~\ref{prob:multi}), not the single-trigger class of this section.
\end{rem}

\begin{exa}
For $P_0=\{e:\varphi_e(0)\!\downarrow\}\in\Sig1$ only polarity A is possible, and $\Lam_E(P_0)$ is $\Sig1$-complete. For $P_1=\{e:\varphi_e(0)\!\uparrow\}\in\Pii1$ only B is possible, realised by augmentation ($\theta=\{(0,0)\}$), and $\Lam_E(P_1)$ is $\Pii1$-complete. For $\TOT$, polarity A gives $\TOT\cap S_a$ and B gives $\TOT\setminus S_a$, both $\Pii2$ and non-recursive.
\end{exa}

\subsection{Topological remark: finite evidence and monitors}\label{sec:topo}

Fix a universal machine $M_U$ and, for an index $x$ and an input $z$, the \emph{trace} of $\varphi_x(z)$: the sequence of configurations of $M_U$ when simulating $x$ on $z$, starting from the initial configuration, which contains the code of $x$ (two different indices of the same function have different traces from the first element on). Let $\mathrm{Conf}^{\le\omega}$ be the trace space and, for a finite trace $\eta$, $[\eta]$ the cylinder of traces extending it. Cylinders form a basis of clopen sets; a \emph{safety property} in the sense of Alpern--Schneider \cite{AS85} is a closed set of this space, and its violation is an open set: a trace violates safety iff some finite prefix of it already violates it.

\begin{prop}[monitors semidecide violation]\label{prop:mon}
Let $M$ be a monitor: a computable procedure that reads incrementally the trace of a base computation ---that of $\varphi_x$ on a fixed input, or the interleaving of the traces on all inputs--- and may, at some point, raise an alarm. The set of traces on which $M$ raises an alarm is open (a union of cylinders), and the set of indices $x$ such that $M$ raises an alarm on the trace of $\varphi_x$ is $\Sig1$. Conversely, every $\Sig1$ set of indices is the alarm set of some monitor: the one that reads $x$ from the initial configuration, simulates $\varphi_a(x)$ and raises an alarm when it halts. (A monitor without access to the initial configuration ---only to the transitions--- would be extensional in the trace and could not be anchored on an inert family of identical traces; this is the trace version of Corollary~\ref{cor:ext}(ii).)
\end{prop}

\begin{proof}
If $M$ raises an alarm on a trace, it does so after reading a finite prefix $\eta$ of it, and raises it equally on every trace in $[\eta]$; the alarm set is $\bigcup[\eta]$ over the $\eta$ that trigger it. The index set is $\{x:\exists s\,(\text{the prefix of length } s \text{ of the trace of } \varphi_x \text{ triggers an alarm})\}$, an $\exists$ over a decidable condition. The converse is Kleene's normal form: $S_a=\{x:\exists s\,T(a,x,s)\}$.
\end{proof}

\begin{rem}[relation to enforceable policies]
Proposition~\ref{prop:mon} is the index-set counterpart of the characterisation of the policies enforceable by execution monitoring as exactly the co-r.e.\ ones (Viswanathan \cite{Vis00}; Hamlen, Morrisett and Schneider \cite[Theorem~3.2.1]{HMS06}). As there, the monitor is given access to the program text (their event $e_M$); without that access the monitor is extensional in the trace and no anchor can be built (Corollary~\ref{cor:ext}(ii)).
\end{rem}

\begin{rem}[classical triggers]
Every $\Sig1$-complete set in mathematics is, by Proposition~\ref{prop:mon}, a possible trigger: the set of Diophantine equations with a solution (Matiyasevich), the word problem in a finitely presented group with unsolvable word problem (Novikov--Boone), provability in an essentially undecidable theory. They are examples of $S_a$, not of operators: they have no base system and no inert wrapper, and their undecidability is imported. The framework does not prove them; it only places them as possible ``engines'' of an operator that someone might build on them.
\end{rem}

\begin{cor}[what the monitor achieves]\label{cor:opt}
Let $E\in\mathcal{E}_{\mathrm B}(P)$ with trigger $a$. Violation ($S_a$) is semidecidable and preservation ($P\setminus S_a$) is not. The canonical monitor ---simulate $\varphi_a(x)$ and raise an alarm when it halts--- semidecides violation; no computable procedure semidecides preservation for \emph{all} the systems that preserve it (a procedure can certify an r.e.\ subset of them, for instance those with a proof in a fixed system, but never the totality).
\end{cor}
\begin{proof}
$S_a$ is r.e.\ and the canonical monitor semidecides it (Proposition~\ref{prop:mon}). $P\setminus S_a$ is not r.e.\ by Corollary~\ref{cor:exp}.
\end{proof}

\begin{rem}
No alternative architecture semidecides preservation for all preserving systems ---not because simulation is the only possible anatomy, but because what a procedure can certify on finite evidence is $\Sig1$, and the untriggered region is not. Architectures that certify more than the monitor (a monitor run alongside a proof search, say) certify r.e.\ subsets of the untriggered region, never the whole. This is the dynamic version of Schneider's characterisation \cite{Sch00} of the policies enforceable at run time.
\end{rem}

\begin{cor}[the extensional/intensional line]\label{cor:ext}
\begin{enumerate}[label=(\roman*),leftmargin=2.2em]
\item Every $E\in\mathcal{E}(P)$ is intensional, and so is its trigger: there are $x,y$ with $\varphi_x=\varphi_y$ and $\Str_a(x)\wedge\neg\Str_a(y)$.
\item If a trigger $a$ is extensional ($\varphi_x=\varphi_y\Rightarrow(\Str_a(x)\Leftrightarrow\Str_a(y))$), then $S_a$ is the index set of an r.e.\ behavioural property $\widetilde S$, and $\Str_a(x)$ iff some finite restriction of $\varphi_x$ is in $\widetilde S$: an extensional trigger can only detect a finite piece of the function, and can never be anchored to $K$ on an inert family.
\end{enumerate}
\end{cor}
\begin{proof}
(i) By (W1), $\varphi_{w(e)}=\varphi_{w(e')}$ for all $e,e'$; by (W2), $\Str_a(w(e))$ iff $e\in K$; take $e\in K$, $e'\notin K$. (ii) $S_a$ is r.e.\ and, by hypothesis, closed under equality of functions: it is the index set of a behavioural $\widetilde S$ with r.e.\ index set; Rice--Shapiro gives the characterisation by finite restrictions. On an inert family all the $w(e)$ compute the same function, so $\Str_a(w(e))$ is constant in $e$ and cannot express $e\in K$.
\end{proof}

\begin{rem}
The question that elevation poses is extensional ---what will the system do---; the trigger that answers it is, necessarily, a search over the letter of the code. This is what Myhill--Shepherdson \cite{MS55} (and its analogue for total functions, Kreisel--Lacombe--Shoenfield \cite{KLS57}) leaves to an extensional operation (detecting a finite restriction) and what the anchor asks for beyond it.
\end{rem}

\subsection{Instances: axes and catalogue}\label{sec:axes}

\begin{defi}[axis, instance]\label{def:eix}
An \emph{observation} is a map ob on $\PC$ (for instance $\mathrm{ob}(g)=g$, $\mathrm{dom}\,g$, $\mathrm{ran}\,g$, the discrepancy set with a fixed reference). A property is \emph{of axis} ob if $\Pt=\mathrm{ob}^{-1}(Q)$ for some $Q$: it depends on behaviour only through ob. An \emph{instance} of the axis is a tuple $(\mathrm{ob},Q,t,a,\pi,\text{base})$ satisfying Definition~\ref{def:oes} and the anchor axiom.
\end{defi}

The four axes of this section are the observations $g$, $\mathrm{dom}\,g$, $\mathrm{ran}\,g$ and the discrepancy set $\{z:g(z)\!\downarrow\ne\varphi_r(z)\}$ with a total reference $\varphi_r$. Proposition~\ref{prop:RS} does not see ob: it only sees the arithmetical class of $P$. Hence the axis decides \emph{what} is governed; the class of $P$ constrains the polarity and the base; and the normal form $(a,\pi)$ is the same for all. ``Rich in axes, unique in form'' means exactly this.

\subsubsection{Functional axis: instrumented rewriting}

\begin{prop}\label{prop:func}
Let $\Pt$ be non-trivial with $\emptyfn\notin\Pt$, and let $\Phi\in\Rinst(P)$ (Section~\ref{sec:rinst}, with the synthesiser producing $\varphi_{\mathrm{inj}(t,e)}=\varphi_t$ if $e\in K$ and $=\emptyfn$ otherwise). Then $\Phi$ is a semantic elevation operator over $P$ of polarity A, with trigger $a$ such that $\varphi_a(x)\simeq\varphi_{\delta(x)}(\delta(x))$ and with the generator $w$ of $\Rinst$ as wrapper. In particular Lemma~\ref{lem:red} and Theorem~\ref{thm:rinst} are the case $\pi=\mathrm{A}$ of Theorem~\ref{thm:main}(iii).
\end{prop}

\begin{proof}
$a$ exists by s-m-n; $\Str_a(x)\iff\varphi_{\delta(x)}(\delta(x))\!\downarrow$. By definition of $\Rinst$, $\varphi_{\Phi(x)}=\varphi_{\mathrm{inj}(t,\delta(x))}$, which is $\varphi_t$ if $\Str_a(x)$ and $\emptyfn$ otherwise: polarity A. (W1): $\varphi_{w(e)}=\varphi_t$. (W2): $\delta(w(e))=e$, so $\Str_a(w(e))\iff\varphi_e(e)\!\downarrow\iff e\in K$.
\end{proof}

\begin{rem}
If $\emptyfn\in\Pt$ (for instance for every $\Pt$ with $P\in\Pii1$), the same class $\Rinst$ is defined with polarity B: the synthesiser produces $\mathrm{mix}_\sigma(t,f)$ for an $f$ satisfying $(\mathrm{H_B})$. A synthesiser that produced exactly $\varphi_f$ when $\varphi_e(e)\!\uparrow$ is not realisable (Lemma~\ref{rem:BB}).
\end{rem}

\subsubsection{Deductive axis: consistency supervision}

We identify a \emph{deductive system} with the r.e.\ set $W_x=\mathrm{dom}\,\varphi_x$ of (codes of) its derivable sentences, and fix a consistent base system $T$ with index $t$ ($W_t=\mathrm{Th}(T)$; $\varphi_t$ is a semidecider of $\mathrm{Th}(T)$). The property is consistency, $\Pt_{\mathrm{con}}=\{g : \ulcorner\bot\urcorner\notin\mathrm{dom}\,g\}$, behavioural (it depends only on the domain), non-trivial, with index set $P_{\mathrm{con}}\in\Pii1$.

\begin{prop}\label{prop:ded}
Let $f$ satisfy $\varphi_f=\varphi_t\cup\{(\ulcorner\bot\urcorner,0)\}$ and $a$ satisfy $\varphi_a(x)\simeq\varphi_{\delta(x)}(\delta(x))$ for a syntactic extractor $\delta$. The canonical operator $E_{\mathrm{ded}}:=E^{\mathrm B}_{t,a,f}$ ---the system that enumerates $W_t$ and, in parallel, executes $\varphi_{\delta(x)}(\delta(x))$, adding $\ulcorner\bot\urcorner$ if it halts--- is a semantic elevation operator over $P_{\mathrm{con}}$ of polarity B by augmentation, with wrapper $w(e)$ = an index of $\varphi_t$ carrying $e$ inertly ($\delta(w(e))=e$). Hence $\Lam_{E_{\mathrm{ded}}}(P_{\mathrm{con}})=P_{\mathrm{con}}\setminus S_a$, $\overline K\mred\Lam_{E_{\mathrm{ded}}}(P_{\mathrm{con}})$, and the preservation of consistency under supervision is $\Pii1$-complete.
\end{prop}

\begin{proof}
Since $T$ is consistent, $\ulcorner\bot\urcorner\notin\mathrm{dom}\,\varphi_t$, and $\varphi_f$ is an augmentation (Lemma~\ref{rem:BB}(ii)), so that for every $\sigma$
\[
\mathrm{mix}_\sigma(t,f)=\varphi_t\cup\{(\ulcorner\bot\urcorner,0)\};
\]
hence $(\mathrm{H_B})$ holds, because this function has $\ulcorner\bot\urcorner$ in its domain. (W1): $\varphi_{w(e)}=\varphi_t$. (W2): $\Str_a(w(e))\iff\varphi_e(e)\!\downarrow$. Theorem~\ref{thm:main}(ii)--(iv) gives the rest; $\Pii1$-completeness because $P_{\mathrm{con}}\in\Pii1$.
\end{proof}

\begin{rem}
This is the hierarchical axis of Section~\ref{sec:sup} with the supervision made explicit: $E_{\mathrm{ded}}(x)$ is the supervisor of the system $x$, and the question ``does the supervised system remain consistent?'' is $\Pii1$-complete. Note that here polarity B is realisable by augmentation because a set of theorems can \emph{grow} after the trigger; on the functional axis the same move is possible only on inputs not yet answered (Lemma~\ref{rem:BB}). We do not claim that any statement of internal provability of the kind of G\"odel's second theorem follows from this: what is obtained is the undecidability of the \emph{family} of supervised systems, not any fact about what a fixed system proves about itself.
\end{rem}

\subsubsection{Relational axis: observational equivalence}

We take as objects pairs $\langle x,y\rangle$ and as property agreement on the common domain,
\[
\Pt_{\mathrm{rel}}=\{\langle x,y\rangle : \forall z\,(\varphi_x(z)\!\downarrow\wedge\varphi_y(z)\!\downarrow\Rightarrow\varphi_x(z)=\varphi_y(z))\},
\]
with index set $P_{\mathrm{rel}}\in\Pii1$. Its dual, \emph{observable discrepancy} $\Disc(x,y):\Leftrightarrow\exists z\,(\varphi_x(z)\!\downarrow\wedge\varphi_y(z)\!\downarrow\wedge\varphi_x(z)\ne\varphi_y(z))$, is $\Sig1$. (We deliberately leave out the case ``one converges and the other does not'', which for each $z$ is $\Sig1\wedge\Pii1$ and would make general non-equivalence $\Sig2$-complete.)

This axis is \emph{not} literally an instance of Definition~\ref{def:oes}: the base is the pair itself (the output outside the trigger is $\langle x,y\rangle$ itself, not a fixed $\varphi_t$), so $E^{-1}(P_{\mathrm{rel}})\ne\Exp(E)$. It is an \emph{operator with its own base}: it satisfies the anchored reduction scheme (inert wrapper plus a $\Sig1$ trigger encoding $K$) and, as we shall see, also the normal form $\Lam=P_{\mathrm{rel}}\setminus S_a$, but only because the trigger is confined to a family on which $(\mathrm{H_B})$ holds pointwise. We prove it directly.

\begin{prop}\label{prop:rel}
Let $\delta$ be a syntactic extractor that takes a fixed value $e_0\notin K$ outside the indices of the form $\mathrm{id}_e$ (the identity with $e$ inert), and $a$ with $\varphi_a(\langle x,y\rangle)\simeq\varphi_{\delta(x)}(\delta(x))$. Let $E_{\mathrm{rel}}(\langle x,y\rangle):=\langle x,y''\rangle$, where $y''$ computes, on $z$, the race between $\varphi_y(z)$ and $\varphi_a(\langle x,y\rangle)$: it returns $\varphi_y(z)$ if $\varphi_y(z)$ finishes first, and $\varphi_x(z)+1$ otherwise. With $w(e):=\langle\mathrm{id}_e,\mathrm{id}\rangle$: (W1) $w(e)\in P_{\mathrm{rel}}$ for every $e$; (W2) $\Str_a(w(e))\iff e\in K$; and $w(e)\in\Lam_{E_{\mathrm{rel}}}(P_{\mathrm{rel}})\iff e\notin K$. Hence $\overline K\mred\Lam_{E_{\mathrm{rel}}}(P_{\mathrm{rel}})$. Moreover: (a) $\Lam_{E_{\mathrm{rel}}}(P_{\mathrm{rel}})=P_{\mathrm{rel}}\setminus S_a$ exactly, where $S_a=\{\langle\mathrm{id}_e,y\rangle:e\in K\}$; (b) $\Lam_{E_{\mathrm{rel}}}(P_{\mathrm{rel}})$ is $\Pii1$-complete; (c) $\Disc$ is $\Sig1$-complete.
\end{prop}

\begin{proof}
$E_{\mathrm{rel}}$ is total computable by s-m-n. (W1): both components of $w(e)$ compute $\mathrm{id}$. (W2): $\delta(\mathrm{id}_e)=e$. If $e\notin K$, the race is always won by $\varphi_y$ and $y''=y$: $E(w(e))=w(e)\in P_{\mathrm{rel}}$. If $e\in K$ with $\sigma=\sigma_a(w(e))$, then $\varphi_{y''}(z)=z+1$ for every $z$ with $\tau_{\mathrm{id}}(z)\ge\sigma$, which exists because $\{z:\tau_{\mathrm{id}}(z)<\sigma\}$ is finite; at such a $z$, $\varphi_{\mathrm{id}_e}(z)=z\ne z+1$, and $E(w(e))\notin P_{\mathrm{rel}}$. (a): since $\delta\equiv e_0\notin K$ outside the $\mathrm{id}_e$, $S_a=\{\langle\mathrm{id}_e,y\rangle:e\in K\}$; for such a pair and \emph{any} $y$, the previous argument (which only uses the first component, $\mathrm{id}_e$, of infinite domain) gives $E(\langle\mathrm{id}_e,y\rangle)\notin P_{\mathrm{rel}}$; outside $S_a$, $y''=y$ and the pair does not change. Hence $\Lam=P_{\mathrm{rel}}\setminus S_a$. (b): $P_{\mathrm{rel}}\setminus S_a$ is $\Pii1$ and $\overline K$-hard. (c): $\Disc(\mathrm{id}_e,y'')\iff e\in K$ gives $K\mred\Disc$; since $\Disc$ is r.e., it is $\Sig1$-complete.
\end{proof}

\begin{rem}[own base and confinement of the trigger]
The axis is different ---it relates two systems, without iterating one or supervising one--- and the reduction scheme is the same as in the others: inert wrapper, $\Sig1$ trigger anchored to $K$, exit from the property only on inputs not yet answered (irrevocability, Lemma~\ref{rem:BB}). The normal form $P_{\mathrm{rel}}\setminus S_a$ holds, but for a reason worth understanding: for an operator with its own base, $(\mathrm{H_B})$ is a \emph{pointwise} condition on each $x$ (that $\varphi_x$ still has unanswered inputs when the trigger fires), and the form holds exactly when it is met on all of $P\cap S_a$. With the $\delta$ of the proposition the trigger fires only on the $\mathrm{id}_e$, of infinite domain, and the condition is met; with an extractor reading $e$ from \emph{any} index, an $x$ computing nothing with $e\in K$ inert would give $\langle x,y\rangle\in P_{\mathrm{rel}}\cap S_a$ with $E(\langle x,y\rangle)\in P_{\mathrm{rel}}$, and the form would fail. The lesson is one of design: the normal form of an operator with its own base is a property of the \emph{confinement of the trigger}, not of the axis; the designer holds it in hand. Fixing the second component turns the axis into conformance (Proposition~\ref{prop:conf}), which is a literal instance.
\end{rem}

\begin{rem}[what the framework tells whoever compares two systems]
Comparing two systems that evolve together (differential testing, shadow deployment, checking a distilled model against the original) inherits all the hardness: preservation of agreement is $\Pii1$-complete and discrepancy is $\Sig1$-complete, so the differential monitor ---execute both and compare--- semidecides discrepancy, and nothing semidecides agreement (Corollary~\ref{cor:opt}). The form ``preserving $=$ in agreement and not triggered'' holds, but only if the trigger is confined to systems that can still disagree; on a system that no longer answers, the trigger may fire without agreement ever being broken, and then the formula describes nothing. Comparing against a \emph{fixed} reference ---an executable specification, a frozen version--- removes the condition: the case enters the class (Proposition~\ref{prop:conf}) and the form holds with no confinement. This is what good comparison practice does, and now we know why.
\end{rem}

\subsubsection{Two more axes, and a catalogue}

The template of the two previous axes ---a $\Pii1$ safety property, polarity B, irrevocable exit--- gives literal instances wherever there is a downward-closed property; we show two more and collect the catalogue.

\begin{prop}[conformance to a reference]\label{prop:conf}
Let $r$ be an index of a \emph{total} function $\varphi_r$ and $\Pt_r:=\{g : \forall z\,(g(z)\!\downarrow\Rightarrow g(z)=\varphi_r(z))\}$ (``no discrepancy with the reference''; since $\varphi_r$ is total, this is $g\subseteq\varphi_r$), behavioural, non-trivial, with $P_r\in\Pii1$. With $t:=r$, $f$ with $\varphi_f=\varphi_r+1$ and $a$ a trigger anchored by an extractor $\delta$, the canonical operator $E^{\mathrm B}_{r,a,f}$ is a semantic elevation operator over $P_r$, and $\Lam(P_r)=P_r\setminus S_a$ is $\Pii1$-complete.
\end{prop}
\begin{proof}
$P_r=\{x:\forall z,s,v\,(\varphi_{x,s}(z)=v\Rightarrow v=\varphi_r(z))\}$ is $\Pii1$ because $\varphi_r$ is total (its graph is recursive; if $\varphi_r$ were partial with non-recursive graph, $P_r$ could be $\Pii2$-complete: for $\varphi_r=0\restr K$ it is). $\mathrm{mix}_\sigma(r,f)$ agrees with $\varphi_r$ on $\{\tau_r<\sigma\}$, which is finite, and equals $\varphi_r+1$ on the rest, which is non-empty: it is not a subfunction of $\varphi_r$, and $(\mathrm{H_B})$ holds. (W1), (W2) with $w(e)$ an index of $\varphi_r$ carrying $e$ inertly (Lemma~\ref{lem:inert}).
\end{proof}

This is the relational axis with the second component fixed: where $E_{\mathrm{rel}}$ has the normal form only by confinement of the trigger (Proposition~\ref{prop:rel}), conformance to a fixed, total reference has it as a literal instance.

\begin{prop}[monitoring a forbidden action]\label{prop:mon2}
Let $\rho$ be a value (``forbidden action'') and $\Pt_\rho:=\{g : \rho\notin\mathrm{ran}\,g\}$, behavioural, non-trivial, $P_\rho\in\Pii1$. With $t$ any index with $\rho\notin\mathrm{ran}\,\varphi_t$, $f$ the constant $\rho$ and $a$ anchored, $E^{\mathrm B}_{t,a,f}$ is a semantic elevation operator over $P_\rho$, and $\Lam(P_\rho)$ is $\Pii1$-complete.
\end{prop}
\begin{proof}
$\{z:\tau_t(z)\ge\sigma\}$ is cofinite (it includes the $z\notin\mathrm{dom}\,\varphi_t$), and there $\mathrm{mix}_\sigma(t,f)$ takes the value $\rho$: $(\mathrm{H_B})$. The rest is identical.
\end{proof}

This is the \emph{watchdog} of Remark~\ref{rem:guia}, now as a behavioural property (the program ``does'' the action by returning it), without going through traces.

\begin{table}[htbp]
\centering\footnotesize
\begin{tabular}{@{}>{\raggedright\arraybackslash}p{2.6cm}>{\raggedright\arraybackslash}p{2.6cm}>{\centering\arraybackslash}p{1.0cm}>{\centering\arraybackslash}p{1.0cm}>{\raggedright\arraybackslash}p{2.1cm}>{\raggedright\arraybackslash}p{3.7cm}@{}}
\toprule
Axis (observation) / exemplar & $\Pt$ & class of $P$ & $\pi$ & base & why $(\mathrm{H}_\pi)$ holds \\
\midrule
Functional, ob$=g$ (Prop.~\ref{prop:func}) & $\{\mathrm{id}\}$ or any with $\emptyfn\notin\Pt$ & $\Pii2$ & A & $g=\emptyfn$ & $\emptyfn\notin\Pt$ \\
Deductive, ob$=\mathrm{dom}\,g$ (Prop.~\ref{prop:ded}) & $\ulcorner\bot\urcorner\notin\mathrm{dom}\,g$ & $\Pii1$ & B & $f=t\cup\{(\bot,0)\}$ & augmentation: $\bot$ enters the domain \\
Conformance, ob$=\{z: g(z)\!\downarrow\ne\varphi_r(z)\}$ (Prop.~\ref{prop:conf}) & $g\subseteq\varphi_r$, $\varphi_r$ total & $\Pii1$ & B & $f=r+1$ & discrepancy on late inputs \\
Monitoring, ob$=\mathrm{ran}\,g$ (Prop.~\ref{prop:mon2}) & $\rho\notin\mathrm{ran}\,g$ & $\Pii1$ & B & $f\equiv\rho$ & $\rho$ appears on late inputs \\
\midrule
Liveness: $P_0$ & $g(0)\!\downarrow$ & $\Sig1$ & A only & $g=\emptyfn$ & $\emptyfn(0)\!\uparrow$; B impossible \\
Safety: $P_1$ & $g(0)\!\uparrow$ & $\Pii1$ & B only & $f=t\cup\{(0,0)\}$ & augmentation; A impossible \\
Totality: $\TOT$ & $g$ total & $\Pii2$ & A and B & $g=\emptyfn$; $f=\emptyfn$ & $\emptyfn$ not total; truncation not total \\
\bottomrule
\end{tabular}
\caption{Catalogue of literal instances of Definition~\ref{def:oes}. The four axes differ in the object governed (a self-rewriting program, a deductive system, a system with respect to a reference, a system with respect to a forbidden action); the three lower exemplars show how the arithmetical class of $\Pt$ constrains the polarity. The regularity is that of Rice--Shapiro: every safety property ($\Pii1$, downward closed) admits B by augmentation or by a forbidden value; every property with a base $g\subseteq t$ outside $\Pt$ admits A.}
\label{tab:cataleg}
\end{table}

\subsection{Usage guide}

\begin{rem}[usage guide]\label{rem:guia}
To check that a concrete mechanism belongs to $\mathcal{E}(P)$, and to read off the consequences, the procedure is always the same:
\begin{enumerate}[label=(\arabic*),leftmargin=2.2em]
\item Identify the base system, the property $\Pt$ and a witness $t$ with $\varphi_t\in\Pt$.
\item Identify the trigger $a$ and check that it is $\Sig1$: a halting condition of a subcomputation extracted from the code.
\item Fix the polarity and check its realisability: $(\mathrm{H_A})$ or $(\mathrm{H_B})$ against the arithmetical class of $P$ (Proposition~\ref{prop:RS}).
\item Build the inert wrapper $w$ and verify (W2). \emph{This step is the reduction}: it is not saved; what the guide contributes is that it always has the same form.
\item Read off Theorem~\ref{thm:main}: the elevated property is $P\cap S_a$ or $P\setminus S_a$; $K$ or $\overline K$ reduces to it; the complexity is, as an upper bound, that of $P$ intersected with $\Sig1$ or $\Pii1$ (it may be lower: with $P=\TOT$ and a trigger confined to the wrappers, $\Lam_E(\TOT)$ is $\Sig1$-complete); the untriggered region is not r.e.
\end{enumerate}
Steps (1)--(3) contain the whole decision: the direction in which the mechanism can look is constrained by the class of the property, not chosen by the engineer, and the base is fixed by irrevocability (what can be lost without retracting anything). Step (4) is Lemma~\ref{lem:inert}. If there is no fixed witness $t$ ---because what must be preserved depends on the input, as in the equivalence between two arbitrary systems--- the wrapper and the anchor may exist all the same and the mechanism inherits the hardness (iii), but the form (ii) holds only if the trigger is confined (Remark~\ref{rem:fora}, Proposition~\ref{prop:rel}).
\emph{Example.} A resource \emph{watchdog}: the base system executes and a monitor halts it if it ever accesses a forbidden resource $\rho$. (1) Base: any program; $\Pt$ = ``never accesses $\rho$'' (a trace property, closed in the sense of Alpern--Schneider), $t$ = a program that does not access it. (2) Trigger: ``the trace contains an access to $\rho$'', $\Sig1$. (3) Polarity B by augmentation (accessing $\rho$ is doing something forbidden, irrevocable); $P$ is $\Pii1$, so A is impossible and B is the right one. (Here $\Pt$ is a property of the trace, not of the computed function; the framework of Section~\ref{sec:prelim} is behavioural, but the degree argument of Proposition~\ref{prop:RS} only uses that $P$ is a $\Pii1$ index set and (W2), and it holds all the same.) (4) Wrapper: $w(e)$ = $t$ with $e$ inert; $E(w(e))$ = program that executes $t$ and, in parallel, simulates $\varphi_e(e)$ and accesses $\rho$ if it halts. All the $w(e)$ have the behaviour of $t$; $E(w(e))$ accesses $\rho$ iff $e\in K$. (5) The property ``the supervised system never accesses $\rho$'' is $P\setminus S_a$, $\Pii1$-complete; the monitor semidecides the violation and no procedure enumerates the preserving systems. What the watchdog can do is halt the system \emph{after} the access; by Corollary~\ref{cor:opt}, no computable mechanism certifies, before the fact, all the systems that will never access $\rho$.
\end{rem}

\subsection{Final remarks}\label{sec:final}

\begin{rem}[amortised diagonalisation]\label{rem:amort}
The framework does not avoid diagonalisation: the anchor axiom \emph{is} a reduction from $K$, and all the self-referential weight lives in the s-m-n reduction of Lemma~\ref{lem:red}, the same as Rice's. What it contributes is that the reduction is always the same ---the wrapper $w$--- and that, once (W1) and (W2) are verified for a mechanism, the complexity of the elevated property is read off Theorem~\ref{thm:main} with no further construction. ``Economy'' means diagonalisation done once for the whole class, not the absence of diagonalisation.
\end{rem}

\begin{rem}[composition and multiple triggers]\label{rem:ershov}
$\mathcal{E}(P)$ is not closed under composition, and composition adds no complexity. If $E_1,E_2\in\mathcal{E}(P)$, then $(E_2\circ E_1)^{-1}(P)=E_1^{-1}(E_2^{-1}(P))$ is the computable preimage of $S_{a_2}$ or of its complement: always $\Sig1$ or $\Pii1$, and only the outer polarity counts. The anchor, however, is not transmitted, and the class is not closed: with canonical operators and constant extractors ($\delta\equiv e_0\notin K$) outside the wrappers, $E_1(w(e))$ falls outside the range of the wrapper of $E_2$ and $(E_2\circ E_1)^{-1}(P)=\emptyset$, which admits no anchor. A different reading ---\emph{persistence along a trajectory} of $k$ steps, $P\cap E_1^{-1}(P)\cap(E_2\circ E_1)^{-1}(P)\cap\cdots$--- gives $P$ intersected with a finite intersection of $\Sig1$ and $\Pii1$ sets; the part contributed by the triggers is a d-r.e.\ set (difference of two r.e.\ sets) for every $k$: level $2$ of the Ershov hierarchy, never higher, because no conjunction of steps produces the union that level $3$ needs ($P$ adds only its own class). What does reach level $k$ is \emph{a single} operator with $k$ nested triggers over a $P$ admitting re-entry by truncation: for $\TOT$, let $\varphi_{E(x)}(z)$ be defined at the first stage $s\ge z$ at which the number of triggers fired is odd; then $E^{-1}(\TOT)$ is the set of $x$ with odd final count, at level $k$ of the Ershov hierarchy. By Proposition~\ref{prop:upper}, levels $\ge3$ require $P\notin\Sig1\cup\Pii1$ (level $2$ is already reached by a co-r.e.\ non-r.e.\ set, such as $E_{\mathrm{ded}}^{-1}(P_{\mathrm{con}})$). The closure properties that do hold are those of triggers: $\Sig1$ is closed under $\wedge$, $\vee$ and composition with computable functions.
\end{rem}

\begin{rem}[weaker anchors]\label{rem:weak}
(W2) requires coupling with $K$. A broader version would ask only $\Str_a(w(e))\iff e\in A$ for a fixed non-recursive r.e.\ set $A$. By the Friedberg--Muchnik theorem there are r.e.\ degrees strictly between $\mathbf 0$ and $\mathbf 0'$, and the version with $A$ would admit operators that the version with $K$ excludes. The price is that Theorem~\ref{thm:main}(i),(iii) weakens to ``$A\mred S_a$'' and ``$\Lam_E(P)$ non-recursive'', without completeness. All the axes of this article are anchored to $K$; we leave open whether some natural axis requires an intermediate $A$ (Section~\ref{sec:open}).
\end{rem}

\begin{cor}[relativisation]\label{cor:rel}
Let $X\subseteq\N$ be an oracle. Replacing $K$ by $K^X$, triggers by $\Sigma^{0,X}_1$ predicates and computable functions by $X$-computable ones, Definition~\ref{def:oes}, Theorem~\ref{thm:main} and Proposition~\ref{prop:RS} hold verbatim, with completeness relative to $K^X$.
\end{cor}
\begin{proof}
All the proofs use only s-m-n, the predicate $T$, the non-recursiveness of $K$ and Rice--Shapiro, which relativise.
\end{proof}

\begin{rem}[moving the bound]
The Turing jumps $\mathbf 0^{(\alpha)}$ move the $\Sig1$ bound of the trigger up the hierarchy, and each level reproduces the same form; that they are the only way to do so does not follow from the corollary. This is consistent with Section~\ref{sec:closure}: $\omega$-iteration of the elevation takes the property to $\Pii2$ without changing the trigger (Remark~\ref{rem:fora}).
\end{rem}

\section{The Pullback Lemma, the Organising Pattern, and Proofs as Triggers}\label{sec:principle}

\subsection{The pullback lemma: Rice one level up}

The relation between the extensional case (Section~\ref{sec:ext}) and the intensional one (Sections~\ref{sec:operator}--\ref{sec:closure}, Section~\ref{sec:normal}), which Section~\ref{sec:intro} calls \emph{the line of jurisdiction}, has an exact statement.

\begin{lem}[pullback]\label{lem:pullback}
For every total computable $\Phi$ and every property $\Pt$ with index set $P$,
\[
\Lam_\Phi(P)\;=\;P\cap\Phi^{-1}(P),
\]
where $\Phi^{-1}(P)=\{x:\Phi(x)\in P\}$ is the property $\Pt$ \emph{transported} along $\Phi$. Then:
\begin{enumerate}[label=(\roman*),leftmargin=2.2em]
\item if $\Phi$ is extensional, $\Phi^{-1}(P)$ is the index set of a behavioural property, and if this is non-trivial, $\Lam_\Phi(P)$ is undecidable by Rice (Proposition~\ref{prop:ext});
\item if $\Phi$ is intensional and uniformly disruptive, $\Phi^{-1}(P)$ is not behavioural, but $K$ or $\overline K$ (according to the polarity of Definition~\ref{def:unif}(ii)) reduces to $\Lam_\Phi(P)$ via the wrapper, and $\Lam_\Phi(P)$ is undecidable by the s-m-n reduction inside the fibre (Lemma~\ref{lem:red}, Theorem~\ref{thm:closure}).
\end{enumerate}
\end{lem}

\begin{proof}
The equality is Definition~\ref{def:one}. (i): if $\varphi_x=\varphi_y$ then $\varphi_{\Phi(x)}=\varphi_{\Phi(y)}$, so $\Phi(x)\in P\iff\Phi(y)\in P$. (ii): by (i) of Definition~\ref{def:unif} $w(e)\in P$, and by (ii) $\Phi(w(e))\in P\iff e\in K$ (resp.\ $e\notin K$); two $w(e),w(e')$ with $e\in K$, $e'\notin K$ compute the same function and have different fates.
\end{proof}

This is the content of ``Rice one level up'', the intuition with which the applied work \cite{Gum26} formulated the problem and which the present article makes precise: Rice always applies \emph{to the target} $P$ ---the question ``does the rewritten program satisfy $P$?'' is a behavioural question about $\Phi(x)$--- and what changes is the \emph{transport} back to $x$. If $\Phi$ looks only at the function, the transport preserves behaviourality and Rice does all the work; if it looks at the letter, the transport destroys it, and Rice's s-m-n reduction has to be repeated at the scale of a single fibre to carry the undecidability of the target down to the letter of the original code. The extensional is the question; the intensional is the vehicle of the answer. The two regimes are not two theorems but the two cases of a single pullback.

\begin{rem}[dynamics as step $0$]
The lemma has a second reading. $\Phi$ acts at compile time; the dynamic question ``will $P$ still hold after the system rewrites itself?'' is therefore the \emph{static} question ``does the program $\Phi(x)$ satisfy $P$?'', a step-$0$ property of the compiled code which, read on $x$, is intensional. One rewriting step adds no dynamics: it absorbs it into a single $\Sig1$ event at run time of the compiled program (clause I). The only dynamics that is not absorbed is unbounded iteration, which adds a universal quantifier in front of the existential one (clause IV, Remark~\ref{rem:fora}).
\end{rem}

\subsection{The organising pattern}

\begin{rem}[Expressivity Principle]\label{rem:principle}
The results above instantiate one pattern, which we state as a remark: it is not a theorem but what the theorems have in common. (I) A computable mechanism that reacts to the behaviour of a system does so through a finite event, and every finite event is of the form $\exists s\,T(a,x,s)$ (Proposition~\ref{prop:mon}). (II) What such a mechanism can certify is an open set of the trace space: the set of systems on which it fires is $\Sig1$ and its complement is $\Pii1$ (Theorem~\ref{thm:main}). (III) When the event is anchored to the system's own dynamics, what lies outside the open set is $K$: the elevated property inherits the degree of $K$ and is not behavioural (Theorem~\ref{thm:main}). (IV) Iterated, one event gives $\Sig1$ or $\Pii1$; infinitely many take a $\Sig1$ property to $\Pii2$ while a $\Pii1$ one stays at $\Pii1$ (Theorem~\ref{thm:pi2}); relativised, each jump reproduces the same form (Corollary~\ref{cor:rel}). In words: the expressivity sufficient for a system to be genuinely general is exactly the expressivity that places it, and any adequate verifier of it, beyond algorithmic certification.
\end{rem}

\begin{rem}[what a restriction buys]\label{rem:negative}
Certifiability is bought by restricting what the transformation can embed, and not by the name of the formalism. Restricting $\Phi$ to a menu that cannot embed any computation ---sanitisers, transformations without an extractor, stasis--- makes it not uniformly disruptive, and the promise question ``given $x\in P$, is $\Phi(x)\in P$?'' may become decidable (Section~\ref{sec:nondisr}). Restricting the rewriter to be primitive recursive buys nothing: $E(x):=\lambda z.[\rho \text{ if } T(\delta(x),\delta(x),z),\ 0 \text{ otherwise}]$ is primitive recursive and total, satisfies $E(w(e))\in P_\rho\iff e\notin K$, and $\Pii1$-completeness survives; nor does it buy termination of the iteration, since the shift $\varphi_x(z+1)$ and the $\Phi$ of Theorem~\ref{thm:pi2} are primitive recursive with $\Pii2$-complete persistence. Strong normalisation in proof theory concerns reduction steps that preserve the computed function: in our terms it is semantic stasis, certifiable because trivial.
\end{rem}

\subsection{Proofs as triggers}

A proof is finite evidence ---``there is a derivation'' is $\Sig1$--- and hence can serve as a trigger (Proposition~\ref{prop:mon}); by Proposition~\ref{prop:RS}, an operator whose trigger is finite evidence can only make systems enter $\Sig1$ properties and only make them leave $\Pii1$ ones. This does not say that no preserving system has a proof that it preserves $P$: for a concrete $e$ with $\varphi_e(e)\!\uparrow$ provable, PA proves that $w(e)$ preserves $P$. It says that no uniform scheme of sound certificates covers the preserving systems. Over programs rather than code--proof pairs, this is the observation of Hamlen, Morrisett and Schneider \cite[Section~5]{HMS06} that the policies certifiable by proof-carrying code are the recursively enumerable ones, applied to the preservation question.

\begin{prop}[no sound theory covers the preserving systems]\label{prop:proof}
Let $E\in\mathcal{E}_{\mathrm B}(P)$, let $T$ be a recursively axiomatised theory, and let $x\mapsto\sigma_x$ be a computable map from indices to sentences of $T$ (the formalisation of ``$x\in\Lam_E(P)$''). Suppose $T$ is \emph{sound} for these sentences: $T\vdash\sigma_x$ implies $x\in\Lam_E(P)$. Then $\mathrm{Prov}_T:=\{x: T\vdash\sigma_x\}$ is r.e., $\mathrm{Prov}_T\subsetneq\Lam_E(P)$, and $\Lam_E(P)\setminus\mathrm{Prov}_T$ is infinite and not r.e. The same holds with the untriggered region $\Exp(E)$ in place of $\Lam_E(P)$.
\end{prop}
\begin{proof}
$\mathrm{Prov}_T$ is r.e.\ because the theorems of $T$ are r.e.\ and $x\mapsto\sigma_x$ is computable; it is contained in $\Lam_E(P)$ by soundness. If $\Lam_E(P)\setminus\mathrm{Prov}_T$ were r.e., then $\Lam_E(P)$ would be the union of two r.e.\ sets, hence r.e., against Corollary~\ref{cor:exp}; for $\Exp(E)$ the same argument uses Theorem~\ref{thm:main}(i). A finite set is r.e., so the difference is infinite, and in particular the inclusion is strict.
\end{proof}

\begin{rem}[reading]
This is G\"odel's first theorem at the scale of the family: for every sound theory there are infinitely many preserving systems whose preservation it does not prove, and they cannot even be listed. On the inert family the gap is explicit: $w(e)$ lies in the uncovered part exactly when $e\notin K$ and $T$ does not prove $\varphi_e(e)\!\uparrow$ (for $T$ able to verify the wrapper construction). Strengthening $T$ moves the frontier and never closes it; how far a transfinite sequence of strengthenings can move it is Problem~\ref{prob:coverage}.
\end{rem}

\section{Open Problems}\label{sec:open}

The results delimit a class ---mechanisms reacting to a single finite event anchored to $K$--- and fix its form and complexity. The problems below sit at the edge of the class.

\subsection{The disruption frontier}

\begin{oprob}\label{prob:front}
Given a total computable $\Phi$ by an index, and a fixed non-trivial $\Pt$, is it decidable whether $\Phi$ is uniformly disruptive with respect to $\Pt$?
\end{oprob}

We conjecture not: the natural upper bound is $\Sig3$ (``there is an index of $w$ such that $w$ is total, $\varphi_{w(e)}=\varphi_t$ for every $e$, and (ii) holds''), or $\Sig4$ if $P\in\Pii2$, and we conjecture completeness for the class. Note that the question cannot be answered by looking at $\Lam_\Phi(P)$ alone: transformations of family (3) in Remark~\ref{rem:tax} have $\Lam_\Phi(P)$ undecidable without being disruptive.

\subsection{Classification by the degree of $\Phi^{-1}(P)$}

\begin{oprob}\label{prob:class}
For a fixed non-trivial $\Pt$, describe the set of degrees
\[
\{\deg_T(\Phi^{-1}(P)) : \Phi \text{ total computable}\}
\]
and, for each degree, the class of transformations attaining it: is every r.e.\ degree $\le_T\deg(P)$ attained, and which natural $\Phi$ fall in each family?
\end{oprob}

For $P\in\Sig1$ each finite iterate $\Lam^{[k]}_\Phi(P)$ is $\Sig1$ and only the $\omega$-limit leaves the class (Theorem~\ref{thm:pi2}); for a single $\Phi$ there is nothing beyond, since $\Lam^\omega_\Phi$ is already a fixed point (Proposition~\ref{prop:coind}), and which families $\langle\Phi_k\rangle$ attain higher jumps is open. Lu \cite{Lu26} proves the corresponding separation for capability rather than for preservation: finite internal self-modification stays inside the layer $\{B : B \le_T A\}$ of its oracle, and stabilised revision is governed by the jump $A'$ via the relativised limit lemma; Proposition~\ref{prop:upper} and Theorem~\ref{thm:pi2} are the analogous statements for the question whether a property is preserved.

\subsection{Several triggers}

\begin{oprob}\label{prob:multi}
Define operators with a sequence $\langle a_k\rangle$ of $\Sig1$ triggers and a polarity for each, and determine whether they admit a normal form analogous to Theorem~\ref{thm:main}: is the elevated property determined by $\langle(a_k,\pi_k)\rangle$, and what hierarchy does the number of triggers (finite, $\omega$, a computable ordinal) induce on the complexity of $E^{-1}(P)$?
\end{oprob}

Remark~\ref{rem:ershov} delimits the answer: composition does not rise above $\Sig1\cup\Pii1$, persistence over finite trajectories does not rise above level $2$ of the Ershov hierarchy, and for $P=\TOT$ every $k$-r.e.\ set is $E^{-1}(\TOT)$ of an operator with $k$ nested triggers. We conjecture that the finite levels of the Ershov hierarchy describe exactly the class with finitely many triggers, and that the $\omega$-iterations of Section~\ref{sec:closure} are the first step outside $\Dl2$.

\subsection{The anchor}

\begin{conj}[naturality of $K$]\label{conj:K}
Every mechanism of semantic governance defined without a priority construction that satisfies (a) and (b') of Corollary~\ref{cor:abc} is anchored to $K$; equivalently, no natural axis requires the generalised anchor of Remark~\ref{rem:weak}.
\end{conj}

The conjecture records that in all the axes examined the trigger reduces from $K$ without any trick, while intermediate r.e.\ degrees have historically required priority methods; one way to make it precise is to require $S_a$ to be creative in Myhill's sense.

\subsection{Two refuted characterisations, and the problem}\label{sec:refuted}

The expressivity criterion has a weak biconditional that is almost a reformulation: a total computable $E$ satisfies $\varphi_{E(x)}\in\Pt\iff\Str_a(x)$ (or $\iff\neg\Str_a(x)$) for some $a$ and every $x$ if and only if $E^{-1}(P)$ is $\Sig1$ (or $\Pii1$). All the content of the framework ---the non-behaviourality of $\Lam_E(P)$, the reduction inside a single fibre $\{x:\varphi_x=\varphi_t\}$--- lives in the inertness of the wrapper, and the natural question is whether inertness has an intrinsic characterisation: \emph{is every $K$-hard intensional operator a semantic elevation operator?} The answer is no, for two reasons worth exhibiting because they delimit what can be expected.

\emph{Passive hardness.} Let $\emptyfn\notin\Pt$, $x_0\in P$ and $c$ an index of $\emptyfn$; set $E(x):=x$ for $x\ne x_0$ and $E(x_0):=c$. Then $\Lam_E(P)=P\setminus\{x_0\}\equiv_m P$, $K$-hard; it is not behavioural ($x_0$ has another index of the same function); but no inert wrapper witnesses the hardness: on every family $w$ with $\varphi_{w(e)}=\psi$ constant, $w^{-1}(\Lam_E(P))$ is $\emptyset$ or $\{e:w(e)\ne x_0\}$, which is decidable. The hardness is inherited from $P$ across the fibres, and the intensional action is decidable fibre by fibre.

\emph{Mixed polarity.} Let $P$ be a property admitting both polarities with base $\emptyfn$ ($\TOT$ or $\{\mathrm{id}\}$: $(\mathrm{H_A})$ and $(\mathrm{H_B})$ hold with $g=f=\emptyfn$). Let $w_{\mathrm A},w_{\mathrm B}$ be two inert wrappers of $t$ with disjoint decidable ranges, the same trigger $a$, and $E$ equal to $E^{\mathrm B}_{t,a,\emptyfn}$ on the range of $w_{\mathrm B}$ and to $E^{\mathrm A}_{t,a,\emptyfn}$ elsewhere. $E$ has an anchor and reacts to a single event, and $E^{-1}(P)$ is a symmetric difference of $S_a$ with a recursive set: neither $\Sig1$ nor $\Pii1$, and $E\notin\mathcal{E}(P)$ by Theorem~\ref{thm:main}(ii). Having an inert wrapper and being $K$-hard is not enough: uniformity of the polarity is needed.

\begin{oprob}[fibre-by-fibre characterisation]\label{prob:inert}
For a total computable $E$ and each fibre $F_g=\{x:\varphi_x=g\}$, let $E_g:F_g\to\N$ be the restriction of $E$. Is the existence of an anchored normal form (Definition~\ref{def:anf}) characterisable in terms of the family $\{E_g^{-1}(P)\cap F_g : g\in\PC\}$? The two examples above refute the two simplest candidates (``$\Lam_E(P)$ is $K$-hard and not behavioural''; ``$E$ has an inert anchor''). A third candidate ---a fibre on which $E_g^{-1}(P)$ is $\Sig1$-complete or $\Pii1$-complete, with the same polarity on every fibre where it is non-trivial--- is refuted too: take $P=\TOT$, $t$ an index of the zero function, $E$ equal to the instrumented rewrite of Section~\ref{sec:rinst} (polarity A) on the range of an inert wrapper $w$ of $t$, and, outside it, $E(x)$ := ``compute $\varphi_x(z)$; diverge if the value is $5$, return it otherwise''. Every fibre other than $F_t$ is mapped entirely into $\TOT$ or entirely out of it, the trace on $F_t$ is $F_t\cap R$ with $R$ r.e.\ and $K$-hard, and yet $\Lam_E(\TOT)$ is neither $\TOT\cap S$ nor $\TOT\setminus S$ for any $\Sig1$ set $S$: on $x_e:=\lambda z.[5 \text{ if } \varphi_e(e) \text{ halts within } z \text{ steps, } 0 \text{ otherwise}]$, outside the range of $w$, membership is $e\notin K$, and on $w(e)$ it is $e\in K$. The polarity is mixed, but hidden in the fibres on which $E$ acts trivially; a characterising condition must therefore constrain those fibres too, and we have neither a candidate that survives this example nor a proof that none exists.
\end{oprob}

Proposition~\ref{prop:complete} already shows that the definition is complete up to the elevated property: every mechanism with an anchored normal form can be replaced by a canonical one. What the problem asks for is different, a characterisation of anchored normal forms by the action of the mechanism itself. The three examples show where the condition has to look: inside each fibre (that the action is not decidable fibre by fibre), across the non-trivial fibres (that the polarity is the same), and across the trivial ones (that the extensional part of $E$ does not carry the opposite polarity).

\subsection{Proof-theoretic strength}

Proposition~\ref{prop:proof} says that each sound r.e.\ theory $T$ certifies an r.e.\ part of the preserving systems and leaves a non-r.e.\ part uncovered, and that strengthening $T$ moves the frontier without closing it. The natural transfinite object is not an iterate of a single transformation but a progression of theories $T_0=T$, $T_{\alpha+1}=T_\alpha$ plus a reflection principle, unions at limits, along ordinal notations \cite{Tur39,Fef62}. Over arbitrary notations the coverage question has no content: for $P\in\Pii1$ the sentences $\sigma_x$ can be taken $\Pii1$, Turing's completeness theorem gives for each $x\in\Lam_E(P)$ a notation $a_x$ for $\omega+1$ with $T_{a_x}\vdash\sigma_x$, and along any single r.e.\ path Proposition~\ref{prop:proof} applies unchanged.

\begin{oprob}\label{prob:coverage}
For $P\in\Pii1$ and an anchor verifiable in elementary arithmetic, is the part of $\Lam_E(P)$ that $T$ certifies exactly the part certified by elementary arithmetic plus consistency iterated, along a natural notation system, below the $\Pii1$-ordinal of $T$ in the sense of Beklemishev \cite{Bek03}? If so, that ordinal ($\varepsilon_0$ for PA) measures how much preservation $T$ certifies.
\end{oprob}

The reflective case ---a system certifying its successor inside the theory it reasons in--- is the L\"obian obstacle of the tiling-agents programme \cite{YH13}; whether it can be stated as an operator with its own base (Section~\ref{sec:axes}) is open. Gentzen's cut elimination is the control case: it preserves the computed function, falls into semantic stasis, and has no anchor.

\subsection{A categorical reading}\label{sec:cat}

The results are recursion-theoretic and invoke no category. In Hyland's effective topos \cite{Hyl82}, where ``decidable'' is ``factors through $2$'' and ``r.e.'' is ``factors through $\Sigma$'', the pullback lemma is literally a pullback, the normal form says that the operators of $\mathcal{E}(P)$ modify $\chi_P$ by a single proposition of $\Sigma$, and the anchor axiom is a pointwise surjection onto $\Sigma$, which is where Lawvere's fixed-point theorem \cite{Law69,Yan03} should appear as the explanation of non-factorisation through $2$. Establishing this ---and the correspondence between the effective non-realisability of $\mathrm{gfp}(F_P)$ (Proposition~\ref{prop:coind}) and the $\Pii2$-completeness of Theorem~\ref{thm:pi2}--- is left as a programme; the recursively defined $\U$ is only included in, not equal to, the categorically defined one, since a simple set does not factor through $2$ and is no $\Lam_\Phi(P)$.

\section*{Acknowledgement}
The structuring, the review of mathematical correctness, the drafting in formal language and the orthotypographic review of this document were carried out with the assistance of the artificial intelligence model Claude (Anthropic, 2026). The definitions, the results and the decisions on content are the author's, who takes full responsibility for them.

\bibliographystyle{alphaurl}
\bibliography{refs}

\appendix
\section{Mechanisms from the literature}\label{app:mech}

The following mechanisms, proposed independently, are members of $\mathcal{E}_{\mathrm A}$ or $\mathcal{E}_{\mathrm B}$ (or monitors in the sense of Proposition~\ref{prop:mon}); the framework proves nothing new about them, but it classifies them and identifies their invariant $(a,\pi)$.
\begin{itemize}[leftmargin=2em]
\item \emph{Runtime-verification monitors} for safety properties (Alpern--Schneider \cite{AS85}, Schneider \cite{Sch00}): $\mathcal{E}_{\mathrm B}$ by augmentation. The base executes unaltered and, on detecting a violating prefix, the system enters an error state; $a$ is the prefix detector.
\item \emph{The G\"odel Machine} \cite{Sch03}: $\mathcal{E}_{\mathrm A}$. The base executes and, when the search finds a proof that the new code is better, it switches; $a$ is $\exists p\,\mathrm{Prf}(p,\text{``improvement''})$. The trigger is $\Sig1$; whether it is moreover anchored to $K$ depends on how the improvement condition is formulated, and we do not claim it in general.
\item \emph{Levin's universal search}: $\mathcal{E}_{\mathrm A}$ with $a$ = ``some candidate passes the test''.
\end{itemize}
None of these remarks replaces an analysis of the mechanisms; they only fix that the property they elevate is $P\cap S_a$ or $P\setminus S_a$, with the complexity of Theorem~\ref{thm:main}.

\begin{table}[htbp]
\centering\footnotesize
\begin{tabular}{@{}>{\raggedright\arraybackslash}p{2.1cm}>{\raggedright\arraybackslash}p{1.6cm}>{\raggedright\arraybackslash}p{2.2cm}>{\raggedright\arraybackslash}p{2.3cm}>{\centering\arraybackslash}p{1.1cm}>{\centering\arraybackslash}p{1.5cm}>{\raggedright\arraybackslash}p{1.8cm}@{}}
\toprule
Mechanism & Base system & Property $P$ (class) & Trigger $S_a$ & $\pi$ & Anchor & Status \\
\midrule
Instrumented rewriting ($\Rinst$, Section~\ref{sec:rinst}, Prop.~\ref{prop:func}) & program $\varphi_x$ & $\varphi_x\in\Pt$, $\emptyfn\notin\Pt$ (e.g.\ $\{\mathrm{id}\}$: $\Pii2$) & $\varphi_{\delta(x)}(\delta(x))\!\downarrow$ & A & yes & instance \\
Consistency supervision (Prop.~\ref{prop:ded}) & deductive system $W_x$ & $\bot\notin W_x$ ($\Pii1$) & $\bot$ appears & B (aug.) & yes & instance \\
Observational equivalence (Prop.~\ref{prop:rel}) & pair $\langle x,y\rangle$ & agreement on common domain ($\Pii1$) & $\varphi_{\delta(x)}(\delta(x))\!\downarrow$, confined to the $\mathrm{id}_e$ & B & yes & own base; form by confinement \\
Limit elevation $\Lam^\omega_\Phi$ (Theorem~\ref{thm:pi2}) & trajectory $\Phi^k(x)$ & $\varphi(0)\!\downarrow$ at every step ($\Sig1$) & infinitely many ($\forall k\exists s$) & --- & --- & outside $\mathcal{E}$: $\Pii2$ \\
\midrule
RV monitor (Alpern and Schneider; Schneider) & program & safety over traces ($\Pii1$) & violating prefix & B (aug.) & construct\-ible & illustration \\
G\"odel Machine & base program & ``proved improvement'' & $\exists p\,\mathrm{Prf}(p,\cdot)$, a $\Sig1$ set & A & not claimed & illustration \\
Levin search & candidates in parallel & ``passes the test'' & a candidate passes & A & not claimed & illustration \\
\midrule
Diophantine (Matiyasevich) & --- & --- & $\exists\vec n\,(p(\vec n)=0)$ & --- & --- & trigger only \\
Word problem (Novikov--Boone) & --- & --- & $w_1=_G w_2$ & --- & --- & trigger only \\
\bottomrule
\end{tabular}
\caption{Mechanisms examined and their status with respect to Definition~\ref{def:oes}. ``Instance'': verified in the text, with inert wrapper, anchor to $K$ and normal form. ``Own base; form by confinement'': inert wrapper and anchor; the normal form holds because the trigger is confined to systems that can leave the property (Proposition~\ref{prop:rel}). ``Illustration'': fits the form, without the anchor having been verified in general. ``Trigger only'': a $\Sig1$-complete set with no base system and no wrapper (remark on classical triggers). Limit elevation appears by contrast: it depends on infinitely many events and has no normal form with a single trigger (Remark~\ref{rem:fora}).}
\label{tab:mecanismes}
\end{table}

\end{document}